\documentclass[final,5p,twocolumn]{elsarticle}%
\usepackage{amssymb, amsmath,accents}
\journal{Nonlinear Analysis: Hybrid Systems}
\usepackage{overpic}
\usepackage{float}

\usepackage{hyperref}  

\usepackage{xcolor}
\usepackage{pdfsync}

\usepackage{graphicx}
\graphicspath{{figures/}}
\usepackage{epsfig}
\usepackage{amsmath}
\usepackage{amssymb}
\usepackage{amsthm}
\usepackage{mathtools}
\usepackage{epstopdf}
\usepackage{amsfonts}
\usepackage{lineno}
\usepackage{tikz}
\usepackage{subcaption}
\usetikzlibrary{positioning,shapes}
\usetikzlibrary{arrows}

\allowdisplaybreaks

\newtheorem{theorem}{Theorem}
\newtheorem{definition}[theorem]{Definition}

\newtheorem{assumption}[theorem]{Assumption}
\newtheorem{proposition}[theorem]{Proposition}
\newtheorem{remark}[theorem]{Remark}

\newcommand{\tm}{\times}%
\newcommand{\trn}{^{\scriptscriptstyle \top}}%

\def\field#1{\mathbb #1}%
\def\R{\field{R}}%
\newcommand{\N}{\mathbb{N}}%
\newcommand{\X}{\mathbb{X}}%
\newcommand{\U}{\mathbb{U}}%
\newcommand{\Y}{\mathbb{Y}}%
\newcommand{\UC}{\mathcal{U}}%
\newcommand{\diag}{\mathrm{diag}}%
\DeclareMathOperator{\id}{id}

\def\K{\mathcal{K}}%
\def\Kinf{\mathcal{K}_\infty}%
\let\ol=\overline%
\let\ul=\underline%

\usepackage{xspace}	
\begin{document}
	
	\begin{frontmatter}
		
		\title{\huge
			Compositional Construction of Abstractions for Infinite Networks of Discrete-Time Switched Systems 
		}
		
		\author[First]{Maryam~Sharifi}
		\address[First]{School of Electrical Engineering and Computer Science, KTH Royal Institute of Technology, Stockholm, Sweden}
		\ead{msharifi@kth.se}
		
		\author[Second]{Abdalla~Swikir}
		\address[Second]{Department of Electrical and Computer Engineering, Technical University of Munich, Germany}
		\ead{abdalla.swikir@tum.de}

		\author[Third]{Navid~Noroozi}
		\address[Third]{Institute for Informatics, LMU Munich, Germany}
		\ead{navid.noroozi@lmu.de}
		
		\author[Third,Fourth]{Majid~Zamani}
		\address[Fourth]{Computer Science Department, University of Colorado Boulder, USA}
		\ead{majid.zamani@colorado.edu}
		
		\cortext[corauth]{The work of N. Noroozi is supported by the DFG through the grant WI 1458/16-1. The work of M. Zamani is supported by the H2020 ERC starting grant AutoCPS (grant agreement No.804639).}
		
\begin{abstract}
In this paper, we develop a compositional scheme for the construction of \emph{continuous} approximations for interconnections of \emph{infinitely} many discrete-time switched systems.
An approximation (also known as abstraction) is itself a continuous-space system, which can be used as a replacement of the original (also known as concrete) system in a controller design process.
Having designed a controller for the abstract system, it is refined to a more detailed one for the concrete system.
We use the notion of so-called simulation functions to quantify the mismatch between the original system and its approximation.
In particular, each subsystem in the concrete network and its corresponding one in the abstract network are related through a notion of \emph{local} simulation functions.
We show that if the local simulation functions satisfy certain small-gain type conditions developed for a network containing infinitely many subsystems, then the aggregation of the individual simulation functions provides an overall simulation function quantifying the error between the overall abstraction network and the concrete one.
In addition, we show that our methodology results in a \emph{scale-free} compositional approach for any finite-but-arbitrarily large networks obtained from truncation of an infinite network.
			We provide a systematic approach to construct local abstractions and simulation functions for networks of linear switched systems.
			The required conditions are expressed in terms of linear matrix inequalities that can be efficiently computed.
			We illustrate the effectiveness of our approach through an application to AC islanded microgirds.
			
		\end{abstract}
	\end{frontmatter}
	
	\section{Introduction}
	Recent technological advances in sensing, computation, and data management have enabled us to  develop smart networked systems providing more \emph{autonomy} and \emph{flexibility}.
	Smart grids, swarm robotics, connected automated vehicles and smart manufacturing are just a few examples of such emerging smart networked systems, in which a large numbers of dispersed agents interact and communicate with each other to achieve a common objective.
	The size and the structure of such networks can be arbitrarily large, time-varying or even unknown, and agents can be \emph{constantly} plugged into and out from the network.
	Emerging control networks necessitate also \emph{sophisticated} control objectives, which go beyond standard goals pursued in classical control theory.
	For instance, a sophisticated objective is to control connected autonomous vehicles merging at a traffic intersection while ensuring safety and fuel economy constraints.
	
	The complexity of control objectives, the large number of participating agents, as well as safety concerns call for automated and provably correct techniques to verify or synthesize controllers for the emerging applications of control systems.
	A promising methodology to address the above issues is achieved by a careful integration of concepts from control theory (e.g. Lyapunov methods and small-gain theory) and those of computer science (e.g. formal methods and assume-guarantee rules)~\cite{belta2017formal,henzinger1998you}.
	Discrete abstractions (a.k.a. symbolic models) is one particular technique to provide automated synthesis of correct-by-design controllers for concrete systems. In this approach, controller synthesis problems can be algorithmically solved over finite abstractions of concrete systems by resorting to automata-theoretic approaches \cite{BLOEM2012911}.
	Then, the constructed controllers can be refined back to the original systems based on some behavioral relations between original systems and their finite abstractions such as approximate alternating simulation relations \cite{pt09} or feedback refinement relations \cite{reissig2016feedback}.
	
	The computational complexity of constructing finite abstractions of the concrete systems makes the practical applicability of these methods considerably challenging.
	Hence, applying such approaches to large-scale systems is not feasible at all.
	An appropriate technique to overcome this challenge is to introduce a pre-processing step by constructing so-called \emph{continuous} abstractions.
	In that way, a continuous-space system, but possibly with a \emph{lower} dimension, is obtained as a substitute of the concrete system~\cite{girard2009hierarchical,lavaei2017compositional,Smith.2018,Smith.2019}.
	We note that the applicability of continuous abstractions is \emph{not} limited to the context of symbolic controllers.
	In fact,  they can be used in other hierarchical control approaches in the lower layers, where a simplified model of the system is used for controller design purposes.
	
	For large-scale networks, it is often more useful to maintain the structure (i.e.\;topology) of the network while abstractions are constructed.
	In that way, corresponding to each participating subsystem of the network, a continuous abstraction is constructed individually. 
	Therefore, the complexity of synthesizing continuous abstractions of infinite-dimensional systems is managed in an efficient way.
	The methodology by which an abstraction for the overall network is achieved via the interconnection of the individual abstractions is called a compositional approach~\cite{rungger2016compositional,ZamaniArcak2017,Noroozi.2018b}.
	In order to guarantee that the aggregation of the individual abstractions provides an abstraction for the overall network, the interaction between subsystems should be weak enough which can be technically described by a small-gain condition~\cite{rungger2016compositional,ZamaniArcak2017,Noroozi.2018b,mallik2018compositional,swikir2019compositional2,swikir2019compositional}.
	
	Small-gain type conditions are intrinsically dependent on the size of the network.
	Hence, one can readily find that the satisfaction of compositionality conditions dramatically degrades as the the number of subsystems increases and may not be valid anymore, see~\cite[Remark 6.1]{ZamaniArcak2017}.
	The works in the literature regarding stability analysis of large-scale systems, e.g.~\cite{NMK21,kawan2019lyapunov,DaP20,Bamieh.2012,Barooah.2009,BPD02}, inspired us to address the scalability issue using an over-approximation of a finite-but-large network with a network composed of \emph{infinitely} many subsystems. We call such aggregated system an infinite network.
	It is widely accepted that an infinite network {captures} the essence of its corresponding finite network; see e.g. a vehicle platooning application in~\cite{Jovanovic.2005b}.
	This treatment leads to an infinite-dimensional system and calls for a more rigorous and detailed setting.
	In particular, we adapt the notion of simulation functions~\cite{girard2009hierarchical} to the case of \emph{infinite}-dimensional switched systems.
	The existence of a simulation function ensures that the error between the output trajectories of the abstract and concrete system is quantitatively bounded in a certain sense (cf. Definition~\ref{def:sim-Function}).
	By exploiting the compositionality approach, we assign an individual simulation function to each subsystem and construct the corresponding local abstraction accordingly.
	Then we aggregate them to construct an abstraction for the overall network.
	We show that if a certain small-gain condition recently developed in~\cite{kawan2019lyapunov} is satisfied, then the aggregation yields a continuous abstraction for the overall concrete network.
	Particularly, for linear networks, our conditions are expressed in terms of linear matrix inequalities, where we \emph{explicitly} construct the individual abstractions as well as the controller refinement formulation.

	Motivated by the scale-dependency issue in the classic compositionality methods, in this paper, a \emph{scale-free} compositional approach for the construction of continuous abstractions for \emph{arbitrarily} large-scale networks of discrete-time switched systems is provided.
	We elucidate the scale-free property of our approach by truncating the infinite network to a finite-but-arbitrary large network and show that the compositional abstraction results are preserved under any truncation.
	To the best of our knowledge, our work is the first to provide a scale-free compositional approach for construction of continuous abstractions.
	In addition to the scalability issue, in a large number of applications, the structure of the network is time-varying in the sense that the communication links between subsystems change over time.
	In power networks, for instance, there exist line switches  and the agents are constantly plugged into and out.
	This calls for considering switched dynamics describing this time dependency of the network structure.
	Our setting, therefore, considers an infinite network of switched systems.
	To validate the effectiveness of our approach, we apply our results to AC microgrids operating in an islanded mode.
	In particular, we show through simulations that the behavior of the network remains \emph{independent} of the size of the network, while the network size dramatically increases.
	
	This paper expands on the conference paper~\cite{Msh}, where uniformity conditions with respect to the switching modes were made. The present work provides a completely non-uniform structure for the simulation functions with respect to the switching signals in the network.
	Moreover, the scale-free property of the result is established, which leads to constructing compositional abstractions for any finite-but-arbitrarily large network.
	Therefore, the current setting allows us to consider more general and realistic scenarios, including the new AC microgrid case study.
	
	
	The rest of the paper is organized as follows. Section \ref{notation} provides the systems description. 
	In Section~\ref{sec:SIM}, we first introduce the notion of simulation functions for switched systems, and then show the importance of the existence of such functions in the construction of abstractions.
	Section~\ref{sec:SIM-cons} contains the main result of the paper, that is the construction of continuous abstractions compositionally using small-gain theory.
	In Section~\ref{linear}, we focus on linear subsystems and provide easier-to-check conditions for the construction of continuous abstractions.
	In Section~\ref{sec:Example}, we apply our results to a network of AC islanded microgrids.
	
	\section{Preliminaries and System Description}\label{notation}
	
	\subsection{Notation} 
	We write $\N_0 (\N)$ for the set of nonnegative (positive) integers.
	For vector norms on finite- and infinite-dimensional vector spaces, we write $|\cdot|$.
	By $\ell^p$, $p\in[1,\infty)$, we denote the Banach space of all real sequences $x = (x_i)_{i\in\N}$ with finite $\ell^p$-norm $|x|_p<\infty$, where $|x|_p = (\sum_{i=1}^{\infty}|x_i|^p)^{1/p}$ for $p < \infty$.
	If $X$ is a Banach space, we write $r(T)$ for the spectral radius of a bounded linear operator $T:X\to X$. 
	The identity function is denoted by $\id$.
	Throughout this work, we will consider $\K$ and $\Kinf$ comparison functions; see~\cite[Chapter 4.4]{Khalil.2002} for definitions.
	
	We consider discrete-time switched subsystems $\Sigma_i$, defined later.
	The arbitrary switching signals are defined as ${\sigma_i}:\N_0 \to S_i$ for each subsystem $\Sigma_i$, ${i\in\N}$, and $S_i=\{1,2,\dots,r_i\}$ is a finite index set with $r_i \in \N$. The set of such switching signals are denoted by ${\cal{S}}_i$. 
	\subsection{Infinite networks}
	First, we define discrete-time switched subsystems which are interconnected to form an infinite network consisting of countably infinite number of control subsystems.
	\begin{definition}\label{def1}
		A discrete-time switched system $\Sigma_i$, $i\in\N$, is defined by the tuple 
		\begin{align}
			\Sigma_i=(\mathbb{X}_i, \mathbb{W}_i, \mathbb{U}_i, \mathcal{U}_i, \mathbb{Y}_i, h_{i,s_i}, f_{i,s_i}, S_i),
		\end{align}
		where $\mathbb{X}_i\subseteq\R^{n_i}$, $\mathbb{W}_i\subseteq\R^{N_i}$, $\mathbb{U}_i\subseteq\R^{m_i}$, and $\mathbb{Y}_i\subseteq\R^{q_i}$ are the state set, internal input set, external input set, and output set, respectively. 
		We use symbol $\mathcal{U}_i$ to denote the set of functions $u_i:\N_0 \rightarrow \mathbb{U}_i$.
		Functions $f_{i,s_i}:\mathbb{X}_i\times\mathbb{W}_i\times\mathbb{U}_i\to\mathbb{X}_i$ are the transition functions for $s_i\in S_i$. Moreover, $h_{i,s_i}:\mathbb{X}_i\to\mathbb{Y}_i$ are the output maps. 
	\end{definition}
	
	The discrete-time switched subsystems $\Sigma_i$, $i\in\N$, are represented by the difference equation of the form%
	\begin{equation}\label{eq_ith_subsystem}
	{\Sigma _i}:\quad \left\{ \begin{array}{l}
	\mathbf x_i(k+1)  = {f_{i,{\sigma_i(k)}}}({\mathbf x_i(k)},{\mathbf w_i(k)},{\mathbf u_i(k)}),\\
	{\mathbf y_i(k)} = h_{i, \sigma_i(k)}(\mathbf x_i(k)),
	\end{array} \right.%
	\end{equation}
	where $\mathbf x_i:\N_0 \rightarrow  \mathbb{X}_i$, $\mathbf w_i:\N_0 \rightarrow  \mathbb{W}_i$, $\mathbf u_i:\N_0 \rightarrow  \mathbb{U}_i$, and $\mathbf{y}_i:\N_0\rightarrow   \mathbb{Y}_i$ are the state signal, internal input signal, external input signal, and output signal, respectively.
	
	The finite set $I_{i,\sigma_i(k)}^{\mathrm{in}} \subset \N \backslash \{i\}$ collects \emph{mode-dependent in-neighbors} of $\Sigma_i$, i.e. systems $\Sigma_j,j \in I_{i,\sigma_i(k)}^{\mathrm{in}}$, \emph{directly} influencing $\Sigma_i$. 
	On the other hand, the finite set $I_{i,\sigma_i(k)}^{\mathrm{out}} \subset \N$, collects \emph{mode-dependent out-neighbors} of $\Sigma_i$, i.e. $\Sigma_{j},j \in I_{i,\sigma_i(k)}^{\mathrm{out}}$, influenced by $\Sigma_i$.
	Note that we assume $i\notin I_{i,\sigma_i(k)}^{\mathrm{in}}\cup I_{i,\sigma_i(k)}^{\mathrm{out}}$, $\forall i\in \N$. 
	The input-output structure of each subsystem $\Sigma_i$, $i\in\N$, is given by
	\begin{subequations}\label{IO}
		\begin{align}
			&\mathbf w_i(k)= \left( {{\mathbf w}_{ij}(k)}\right)_{j\in I_{i,\sigma_i(k)}^{\mathrm{in}}}\in \mathbb{W}_i:= \prod_{j\in ( I_{i,\sigma_i(k)}^{\mathrm{in}})}\mathbb{W}_{ij},\label{IO1}\\
			&\mathbf y_i(k)= \left({{\mathbf y}_{ij}(k)}\right)_{j\in (i\cup I_{i,\sigma_i(k)}^{\mathrm{out}})}\in\mathbb{Y}_i:= \!\!\!\!\prod_{j\in (i\cup I_{i,\sigma_i(k)}^{\mathrm{out}})}\!\!\mathbb{Y}_{ij},\label{IO2}\\
			&{h_{i, \sigma_i(k)}}({\mathbf x_i(k)})=\left( {h_{ij,\sigma_i(k)}(\mathbf x_i(k))}\right)_{j\in (i\cup I_{i,\sigma_i(k)}^{\mathrm{out}})}.\label{IO2}
		\end{align}
	\end{subequations}
	We denote $\mathbf{w}_i(k)$ for $N_i := \sum_{j\in I_{i,\sigma_i(k)}^{\mathrm{in}}}n_j$, as the internal inputs describing the interconnections among subsystems. 
	The outputs $\mathbf{y_{ij}}(k)$, $j\in I_{i,\sigma_i(k)}^{\mathrm{out}}$, are considered as internal outputs which are used to construct interconnections between subsystems, whereas $\mathbf{y_{ii}}(k)\in\mathbb{Y}_{ii}$ are denoted as external outputs.
	Note that $\mathbf w_i(k)$ and $\mathbf y_i(k)$ are partitioned into sub-vectors and we aggregate all the subsystems $\Sigma_i$ through the interconnection constraints given by $ \mathbf w_{ij}(k)=\mathbf y_{ji}(k)$ for all $i\in\N$ and for all $j\in I_{i,\sigma_i(k)}^{\mathrm{in}}$.

	To model the state (resp. input) space of the overall network, we introduce a Banach space of sequences $x = (x_i)_{i\in\N}$ (resp. $u = (u_i)_{i\in\N}$).
	The most natural choice  is the $\ell^p$-space, precisely, defined as follows: we first fix a norm on each $\mathbb{X}_i$;
	then, for every $p \in [1,\infty)$, we put%
	\begin{equation*}
		\ell^p(\N,(n_i)) := \Bigl\{x = (x_i)_{i\in\N} : x_i \in \mathbb{X}_i,\ \sum_{i\in\N}|x_i|^p < \infty \Bigr\},%
	\end{equation*}
	and equip this space with the norm $|x|_p := (\sum_{i\in\N}| x_i|^p)^{1/p}$.
	Now, we provide a formal definition of the infinite network.
	\begin{definition}\label{sigma}
		Consider subsystems $\Sigma_i=(\X_i, \mathbb{W}_i, \U_i, \UC_i, \Y_i, h_{i,s_i}, f_{i,s_i}, S_i)$, $i\in\N$, with the input-output structure as in \eqref{IO}.
		A discrete-time infinite network $\Sigma$ is defined by the tuple $\Sigma=(\X, \U, \UC, \Y, h_s, f_s, S)$, where $\X = \ell^p(\N,(n_i)) \subset \prod_{i\in\N}\X_i$ with a fixed $p \in[0,\infty)$  and $\U = \ell^q(\N,(n_i)) \subset \prod_{i\in\N}\U_i$ with a fixed $q \in[0,\infty)$.
		The space of admissible \emph{external input functions} $\mathbf{u}$ is defined by $\UC := \bigl\{\mathbf{u}:\N_0 \rightarrow \mathbb{U}  \bigr\}$. 
		Moreover, $h_s(x)=(h_{ii,s_i}(x_i))_{i\in\N}, s\in S, S=\prod_{i\in\N}S_i$ denotes the output function, where $h_s:\mathbb{X}\to\mathbb{Y}$, $\mathbb{Y} \subset \prod_{i\in\N}\mathbb{Y}_{ii}$.
		In addition, we restrict $f_s(x,u)=(f_{i,s_i}(x_i,w_i,u_i))_{i\in\N}$ to $f_{s}:\X \tm \U\to\X$.
	\end{definition}
	
		In that way, the interconnection of subsystems $\Sigma _i$, $i\in \N$, is described by 
	\begin{equation}\label{eq_interconnection}
	{\Sigma}:\quad \left\{ \begin{array}{l}
	\mathbf x(k+1)  = {f_{\sigma(k)}}({\mathbf x(k)},{\mathbf u(k)}),\\
	{\mathbf y(k)} = h_{\sigma(k)}(\mathbf x(k)),
	\end{array} \right.%
	\end{equation}
	where $\mathbf x(k) \!\!=\!\! (\mathbf x_i(k))_{i\in\N}$,  $\mathbf u(k) \!\!=\!\! (\mathbf u_i(k))_{i\in\N}$, $\mathbf y(k) = (\mathbf y_{ii}(k))_{i\in\N}$, $\sigma(k) \!\!=\!\! (\sigma_i(k))_{i\in\N}$, $f_{\sigma(k)}(\mathbf x(k),\mathbf u(k)) = \left(f_{i,\sigma_i(k)}(\mathbf x_i(k),\mathbf w_i(k),\mathbf u_i(k))\right)_{i\in\N}$, and $h_{\sigma(k)}(\mathbf x(k)):=\left(h_{ii,\sigma_i(k)}(\mathbf x_i(k))\right)_{i\in\N}$.
	We call the overall system~\eqref{eq_interconnection} an \emph{infinite} network and denote the corresponding solutions by $\mathbf x(k,x,\sigma,\mathbf u)$ for any $k\in\N_0$, any initial value $x\in \X$, any switching signal $\sigma:\N_0\to\cal{S}$, ${\cal{S}}:=\{\sigma:{\N_0} \to S\}$, and any control input $\mathbf u\in\UC$.
	
	We refer to system \eqref{eq_interconnection} as the \emph{concrete} system, which is often hard to control or analyze. 
	To simplify the controller design process, we, instead, use a simpler and less precise system called an \emph{abstract} system.
	\section{Abstractions for Discrete-Time Switched Systems}\label{sec:SIM}
	
	In this section, we introduce a notion of simulation functions for discrete-time switched systems.
	A simulation function quantifies a relation between the concrete system and its abstraction in the sense that the mismatch between their output trajectories remains bounded (cf. Proposition~\ref{thm:SF}).
	A simulation function is formally defined as follows.
	\begin{definition}\label{def:sim-Function}
		Consider two systems $\Sigma=(\mathbb{X}, \mathbb{U}, \mathcal{U}, \mathbb{Y}, h_s, f_s, S)$ and $\hat \Sigma=(\mathbb{\hat X}, \mathbb{\hat U}, \mathcal{\hat U}, \mathbb{\hat Y}, \hat h_s, \hat f_s, S)$ with the same output space dimensions. Let $p,q \in [1,\infty)$ be given.
		Let $V_s:\X\times \mathbb{\hat X} \rightarrow \R_+, s\in S$, be a family of functions. Assume that there exist constants $\alpha, b >0$ such that for all $s \in S$, all $x \in \X$ and all $\hat x \in \mathbb{\hat X}$,
		\begin{align}\label{sim2}
			&\alpha \left| {h_s(x) - \hat h_s(\hat x)} \right|_p^b \le V_s(x,\hat x),
		\end{align}
		and there exist a function ${\rho _{\rm ext}}\in\K$ and a constant $0<\lambda<1$, such that for all $s',s\in S$
		and all $x \in \X$, $\hat x \in \mathbb{\hat X}$ and $\hat u\in \mathbb{\hat U}$, there exists  $u \in \mathbb{U}$ so that we have
		\begin{align}\label{sim3}
			&\begin{array}{l}
				V_{s'}(f_{s}(x,u),\hat f_{s}(\hat x,\hat u)) - V_{s}(x,\hat x)\\\leq
				- \lambda V_{s}(x,\hat x) + {\rho _{\rm ext}}(|\hat u|_{q}).
			\end{array} %
		\end{align}
		Functions $V_s$ satisfying \eqref{sim2} and~\eqref{sim3} are called simulation functions from $\hat \Sigma$ to $\Sigma$ and $\hat\Sigma$ is called an abstraction of $\Sigma$.%
	\end{definition}
	Now we show that the existence of a simulation function ensures that the output trajectories of the abstract and concrete systems remain within a bounded distance from each other.
	\begin{proposition}\label{thm:SF}
		Consider systems $\Sigma=(\mathbb{X}, \mathbb{U}, \mathcal{U}, \mathbb{Y}, h_s, f_s, S)$ and $\hat \Sigma=(\mathbb{\hat X}, \mathbb{\hat U}, \mathcal{\hat U}, \mathbb{\hat Y}, \hat h_s, \hat f_s, S)$ with the same output space dimensions.
		Let a set of simulation functions $V_s$, $s\in S$, from $\hat\Sigma$ to $\Sigma$ and $p,q \in [1,\infty)$ be given. Then there exist a function $\gamma_{\rm{ext}}\in\K$ and positive constants $\vartheta$ and $\beta<1$, such that for any $\sigma\in\cal{S}$, $x\in \X$, $\hat x \in \mathbb{\hat X}$, $\mathbf{\hat u}\in \hat \UC$, $k\in\N_0$, there exists $\mathbf u\in \UC$ so that we have
		\begin{align}\label{ms}
			&\left| {\mathbf y(k,x,\sigma,\mathbf u) -  \mathbf {\hat y}(k,\hat x,\sigma,\mathbf{\hat u})} \right|_p \nonumber\\&\le \vartheta\beta^k (V_{\sigma(0)}(\xi,\hat \xi))^{\frac{1}{b}}+\gamma_{\rm ext}(|\mathbf{\hat u}|_{q,
				\infty}),
		\end{align}	
		where $|\mathbf{\hat u}|_{q,\infty} := \sup_{k\in \N_0}|\mathbf{\hat u}(k)|_q$ and $b$ as in~\eqref{sim2}.
		
		\begin{proof}
			The proof follows similar arguments as those in the proof of~\cite[Lemma 3.5]{Jiang.2001}.
			Take any $\varepsilon \in (0,\lambda)$ and define ${\cal{D}} := \{(x,\hat x) \in \X \times \mathbb{\hat X}  : V_{s}(x,\hat x) \leq \frac{1}{\lambda -\varepsilon}\rho_{\rm ext}(|\hat u|_{q})\}$ for all $s\in S$.
			It follows from~\eqref{sim3} that
			\begin{align} \label{eq:diss}
				&V_{s'} (f_{s}(x,u),\hat f_{s}(\hat x,\hat u)) - V_{s}(x,\hat x)\nonumber\\ &\leq -\varepsilon V_{s}(x,\hat x) + {\rho_{\rm ext}(|\hat u|_{q}) - (\lambda -\varepsilon	) V_{s}(x,\hat x)} .
			\end{align}
			For all $(x,\hat x) \in \{\X \times \mathbb{\hat X}\} \backslash {\cal{D}}$, we have $V_{s}(x,\hat x) >\frac{1}{\lambda -\varepsilon}\rho_{\rm ext}(|\hat u|_{q})$. Thus, we have
			\begin{align*}
				&V_{s'} (f_{s}(x,u),\hat f_{s}(\hat x,\hat u)) - V_{s}(x,\hat x)\leq -\varepsilon V_{s}(x,\hat x),
			\end{align*}
			for all $(x,\hat x) \in \{\X \times \mathbb{\hat X}\} \backslash {\cal{D}}$, that can be written as
			\begin{align*}
				&V_{\sigma(k)} (\mathbf x(k,x,\sigma,\mathbf u), \mathbf {\hat x}(k,\hat x, \sigma,\mathbf{\hat u}))\\&\leq (1-\varepsilon) V_{\sigma(k-1)}(\mathbf x(k-1,x,\sigma,\mathbf u), \mathbf {\hat x}(k-1,\hat x, \sigma,\mathbf{\hat u})),
			\end{align*}
			for all $k\in \N_0$.
			Therefore, we obtain
			\begin{align} \label{eq:outside-S}
				&V_{\sigma(k)} (\mathbf x(k,x,\sigma,\mathbf u), \mathbf {\hat x}(k,\hat x, \sigma,\mathbf{\hat u}))\leq (1-\varepsilon)^{k} V_{\sigma(0)}(\xi,\hat \xi).
			\end{align}
			Now consider $(x,\hat x)\in {\cal{D}}$.
			It follows from~\eqref{eq:diss} that
			\begin{align}\label{eq:inv}
				&V_{s'} (f_{s}(x,u),\hat f_{s}(\hat x,\hat u)) \leq (1-\lambda) V_{s}(x,\hat x) + \rho_{\rm ext}(|\hat u|_{q})\nonumber\\
				& \leq   \frac{1-\lambda}{\lambda -\varepsilon}\rho_{\rm ext}(|\hat u|_{q})  + \rho_{\rm ext}(|\hat u|_{q})= \frac{1-\varepsilon}{\lambda -\varepsilon}\rho_{\rm ext}(|\hat u|_{q}).
			\end{align}
			Inequalities \eqref{eq:outside-S} and \eqref{eq:inv} imply that
			\begin{align} \label{eq:decay-eISS}
				&V_{\sigma(k)} (\mathbf x(k,x,\sigma,\mathbf u), \mathbf {\hat x}(k,\hat x,\sigma,\mathbf {\hat u}))\nonumber\\&\leq (1-\varepsilon)^{k} V_{\sigma(0)}(x,\hat x)+\frac{(1-\varepsilon)}{\lambda -\varepsilon}\rho_{\rm ext}(|\mathbf {\hat u}|_{q, \infty}).
			\end{align}
			It follows from~\eqref{sim2} and~\eqref{eq:decay-eISS} that
			\begin{align*}
				&\alpha\left| {\mathbf y(k,x,\sigma,\mathbf u) - \mathbf {\hat y}(k,\hat x, \sigma,\mathbf {\hat u})} \right|_p^b \nonumber\\&\leq (1-\varepsilon)^{k} V_{\sigma(0)}(x,\hat x)+\frac{(1-\varepsilon)}{\lambda -\varepsilon}\rho_{\rm ext}(|\mathbf {\hat u}|_{q, \infty}),
			\end{align*}
			which implies that
			\begin{align*}
				&\left| {\mathbf y(k,x,\sigma,\mathbf u) - \mathbf {\hat y}(k,\hat x, \sigma,\mathbf {\hat u})} \right|_p \nonumber\\&\leq \vartheta\beta^k (V_{\sigma(0)}(x,\hat x))^{\frac{1}{b}}+\gamma_{\rm ext}(|\mathbf {\hat u}|_{q, \infty}),%
			\end{align*}
			where $\vartheta=(2\frac{1}{\alpha})^{\frac{1}{b}}$, $\beta=(1-\varepsilon)^{\frac{1}{b}}$, $\gamma_{\rm ext}(\cdot)=(2\frac{1-\varepsilon}{\alpha(\lambda -\varepsilon)}\rho_{\rm ext}(\cdot))^{\frac{1}{b}}$. This completes the proof.
		\end{proof}
	\end{proposition}
	\begin{remark}\label{input}
		Suppose that we are given an interface function $\nu$, which maps
		every $x, \hat{x}$, $\hat{u}$, and $s$ to an input $u = \nu(x, \hat{x},\hat{ u},s)$ so that \eqref{sim3} is
		satisfied. Then, the input $\mathbf{u}$ that realizes \eqref{ms} is readily given
		by $\mathbf{u}(k)=\nu(\mathbf x(k), \hat{\mathbf x}(k),\hat{ \mathbf u}(k),\sigma(k))$; see \cite[Theorem 1]{Runggerhscc}.
	\end{remark}
	Due to the size of the systems, a simulation function from $\hat\Sigma$ to $\Sigma$ is quite hard to be \emph{directly} computed. To address this complexity, we follows a compositional approach and define local simulation functions for each finite-dimensional subsystem (cf. Definition \ref{SIM_vi_existence}). This enables us to verify \eqref{sim2} and \eqref{sim3} in a bottom-up way. The next section develops this strategy with the use of small-gain theory for infinite networks.
	\section{Compositional Construction of Abstractions and Simulation Functions}\label{sec:SIM-cons}
	In the following, we provide a method for compositional construction of simulation functions between the infinite networks $\Sigma$ and $\hat\Sigma$. 
	We assume that each subsystem $\Sigma_i=(\mathbb{X}_i, \mathbb{W}_i, \mathbb{U}_i, \mathcal{U}_i, \mathbb{Y}_i, h_{i,s_i}, f_{i,s_i}, S_i)$ and $\hat \Sigma_i=(\mathbb{\hat X}_i, \mathbb{\hat W}_i, \mathbb{\hat U}_i, \mathcal{\hat U}_i, \mathbb{\hat Y}_i, \hat h_{i,s_i}, \hat f_{i,s_i}, S_i)$ admits a local simulation function as defined below. 
	\begin{definition}\label{SIM_vi_existence}
		Consider subsystems $\Sigma_i=(\mathbb{X}_i, \mathbb{W}_i, \mathbb{U}_i, \mathcal{U}_i, \mathbb{Y}_i, h_{i,s_i}, f_{i,s_i}, S_i)$ and $\hat \Sigma_i=(\mathbb{\hat X}_i, \mathbb{\hat W}_i, \mathbb{\hat U}_i, \mathcal{\hat U}_i, \mathbb{\hat Y}_i, \hat h_{i,s_i}, \hat f_{i,s_i}, S_i)$, $i \in \N$.
		Let $p,q \in [1,\infty)$ be given.
		Assume that there exist functions $V_{i,s_i}:\mathbb{X}_i\times \mathbb{\hat X}_i \rightarrow \R_+, s_i\in S_i$, satisfying the following properties
		\begin{itemize}
			\item There are  constants $\alpha_i >0$ so that for all $x_i \in \mathbb{X}_i$ and all $\hat x_i \in \mathbb{\hat X}_i$
			\begin{equation}\label{sim_eq_viest}
			\alpha_i \left| {h_{i, s_i}(x_i) - \hat h_{i, s_i}(\hat x_i)} \right|^p \le V_{i,s_i}(x_i,\hat x_i).%
			\end{equation}
			\item  There are positive constants $\lambda_i<1,\rho_{i,\rm int}, \rho_{\rm{i,ext}}$ such that for all $ s'_i, s_i\in S_i$, all $x_i \in \mathbb{X}_i$, all $\hat x_i \in \mathbb{\hat X}_i$, all $\hat u_i\in\mathbb{\hat U}_i$, there exists $u_i\in\mathbb{U}_i$ so that the following holds for all $ w_i \in \mathbb{W}_i$ and all $\hat w_i \in \mathbb{\hat W}_i$ 
			\begin{align}\label{sim_eq_nablaviest}	
				\begin{array}{l}
					V_{i, s'_i}\left(f_{i,s_i}(x_i, w_i, u_i),\hat f_{i,s_i}(\hat x_i, \hat w_i, \hat u_i)\right)\\- V_{i, s_i}\left(x_i,\hat x_i\right)\le
					- \lambda_i V_{i, s_i}(x_i,\hat x_i)\\ + {\rho _{\rm{i,ext}}}|\hat u_i|^{q} + {\rho _{{\mathop{i,\rm int}} }}\left| { w_i - {\hat w_i}} \right|^p.
				\end{array}%
			\end{align}	
		\end{itemize}
	\end{definition}
	Then functions $V_{i,s_i}$ are called local simulation functions from $\hat \Sigma_i$ to $\Sigma_i$ and $\hat\Sigma_i$ are called abstractions of $\Sigma_i$ for each $i\in\N$.
	
	Assume that each $\Sigma_i$ admits an abstraction $\hat\Sigma_i$, $\forall i\in\N$, given as in Definition \ref{SIM_vi_existence}. We establish a compositional approach for the construction of continuous abstractions of infinite networks~\eqref{eq_interconnection} by aggregating individual continuous abstractions $\hat\Sigma_i$.
	To do so, we need interaction between subsystems to be sufficiently weak, which is quantitatively described by a small-gain condition, see Assumption~\ref{ass:spectral-radius} below. 
	
	To employ the small-gain theorem, the following conditions are required.
	The first one makes uniformity conditions on the constants given by Definition~\ref{SIM_vi_existence}.%
	
	\begin{assumption}\label{ass_external_gains}
		There are constants $\underline{\alpha}, \underline{\lambda}, {{\overline\rho _{\rm{ext}}}} > 0$ so that for all $i\in \N$, we have $\ul\alpha \leq \alpha_i , \ul\lambda \leq \lambda_i, {\rho _{\rm{i,ext}}} \leq {{\overline\rho _{\rm{ext}}}}$.
	\end{assumption}
	We collect the coefficients from \eqref{sim_eq_viest} and \eqref{sim_eq_nablaviest} to define 
	\begin{align}\label{gamma}
		&{\gamma _{ij}}: = \left\{ {\begin{array}{*{20}{c}}
				{{\rho _{i,{\rm{int}}}}{{\bar N}_i}\frac{1}{{{\alpha _j}}},}& j \in I_{i,s_i}^{\mathrm{in}} ,\\
				{0,}& j \notin I_{i,s_i}^{\mathrm{in}} ,
			\end{array}} \right.
		\end{align}
		where ${\bar N}_i$ denotes the cardinality of the set $I_{i,s_i}^{\mathrm{in}}$.
		\newline
		We additionally introduce the following matrices. 
		\begin{align}\label{condd1}
			&\Lambda := \diag(\lambda_1,\lambda_2,\lambda_3,\ldots), \;\;\Gamma := (\gamma_{ij})_{i,j\in\N}.
		\end{align}
		Now, we define the following matrix by which we express our \emph{small-gain} condition
		\begin{align}
			\Psi := \Lambda^{-1} \Gamma := (\psi_{ij})_{i,j\in\N} , \,\,\, \psi_{ij} = \gamma_{ij}/\lambda_i .
		\end{align}
		We also make an assumption on the boundedness of the operator $\Gamma$.
		\begin{assumption}\label{SIM_def}
			The operator $\Gamma = (\gamma_{ij})_{i,j\in\N}$ satisfies
			$\sup_{j \in \N} \sum_{i=1}^{\infty} \gamma_{ij} < \infty$.
		\end{assumption}	
		Note that Assumption \ref{SIM_def} always holds if each subsystem is interconnected to finitely many subsystems and \emph{no} global communication is used. 
		
		The following spectral radius condition provides a \emph{quantitative} bound on the strength of couplings between the subsystems.
		This is, in fact, the small-gain condition that is required to guarantee that the aggregation of $\hat\Sigma_i$ gives a continuous abstraction for network $\Sigma$.
		\begin{assumption}\label{ass:spectral-radius}
			The spectral radius $r(\Psi) < 1$.
		\end{assumption}
		The following theorem gives the \emph{main result} of the paper, which is a compositional approach to construct the abstractions of infinite interconnected switched control systems and their corresponding simulation functions.
		\begin{theorem}\label{MTC}
			Consider infinite networks $\Sigma=(\mathbb{X}, \mathbb{U}, \mathcal{U}, \mathbb{Y}, h_s, f_s, S)$ and $\hat \Sigma=(\mathbb{\hat X}, \mathbb{\hat U}, \mathcal{\hat U}, \mathbb{\hat Y}, \hat h_s, \hat f_s, S)$. Let $p,q \in [1,\infty)$ be given. Let local simulation functions $V_{i,s_i}:\mathbb{X}_i\times \mathbb{\hat X}_i \rightarrow \R_+, s_i\in S_i$, satisfy Assumptions,~\ref{ass_external_gains},~\ref{SIM_def} and \ref{ass:spectral-radius}.
			Then there exists a vector $\mu = (\mu_{i})_{i\in\N}\in \ell^{\infty}$ satisfying $\underline{\mu} \leq \mu_i \leq \overline{\mu}$ with constants $\underline{\mu},\overline{\mu}>0$ such that the following is satisfied
			\begin{equation}\label{eq_muest2}
			\frac{[\mu\trn(-\Lambda + \Gamma)]_i}{\mu_{i}} \leq -\lambda_{\infty}, \quad \forall i \in \N,
			\end{equation}
			for a constant $\lambda_{\infty} \in(0,1)$.
			Moreover, the following family of functions $V_s:\X\times \mathbb{\hat X} \rightarrow \R_+, s\in S$, with $S=\prod_{i\in\N}S_i$,
			\begin{equation*}\label{eq:sim-function-construction}
				V_s(x, \hat x) = \sum_{i=1}^{\infty} \mu_{i} V_{i,s_i}(x_i, \hat x_i),\quad V_s:X\times \hat X \rightarrow \R_+,  %
			\end{equation*}
			are simulation functions from $\hat \Sigma$ to $\Sigma$ with $b = p , \alpha = \ul\mu	\ul\alpha$ as in~\eqref{sim2} and $\lambda = \lambda_\infty$ and $\rho_{\mathrm{ext}} : t \mapsto \overline{\mu}\; \ol\rho _{\rm{ext}} t^q$ as in~\eqref{sim3}.	
		\end{theorem}
		\begin{proof}
			From~\cite[Lemma V.10]{kawan2019lyapunov}, Assumption~\ref{ass:spectral-radius} (i.e. $r(\Psi) < 1$) implies that there exists a vector $\mu = (\mu_{i})_{i\in\N}\in \ell^{\infty}$ satisfying $\underline{\mu} \leq \mu_i \leq \overline{\mu}$ such that~\eqref{eq_muest2} holds.
			
			Now we show that $V$ in~\eqref{eq:sim-function-construction} satisfies~\eqref{sim2} with $\alpha = \ul\mu	\ul\alpha$.
			For any $s\in S$, $s_i\in S_i$, $x\in X$ and $\hat x\in \hat X$ and taking $b=p$, it follows from~\eqref{sim_eq_viest} and Assumption~\ref{ass_external_gains} that
			\begin{align*}
				\sum_{i=1}^{\infty} \mu_i V_{i, s_i}(x_i, \hat x_i) &\geq \sum_{i=1}^{\infty} \mu_i \alpha_i|h_{i, s_{i}}(x_i) - \hat h_{i, s_{i}}(\hat x_i)|^p\\
				& \geq \ul\mu	\ul\alpha \sum_{i=1}^{\infty} |h_{ii, s_{i}}(x_i) - \hat h_{ii, s_{i}}(\hat x_i)|^p \\
				& \geq \ul\mu	\ul\alpha |h_s(x) - \hat h_s(\hat x)|_p^p  .
			\end{align*}
			Next we show the inequality \eqref{sim3} holds as well. Considering \eqref{sim_eq_nablaviest} and \eqref{eq:sim-function-construction}, we obtain the chain of inequality in \eqref{eq_nablaviest} for all $s'_i, s_i\in S_i$, $s_j\in S_j$, $s', s\in S$, $i\in \N$.
			\begin{figure*}[ht]
				\rule{\textwidth}{0.4pt}
				\begin{align}\label{eq_nablaviest}
					\begin{array}{l}
						V_{s'}\left(f_{s}(x, u),\hat f_{s}(\hat x, \hat u)\right)- V_{s}(x,\hat x) = \sum\limits_{i = 1}^\infty  \mu_{i}\left[{V_{i, s'_i}}\left(f_{i, s_i}(x_i, w_i, u_i),\hat f_{i, s_i}(\hat x_i, \hat w_i, \hat u_i)\right)- {V_{i, s_i}}(x_i,\hat x_i) \right]
						\\\le \sum\limits_{i = 1}^\infty  {\mu_{i}}( - {\lambda_i}{V_{i, s_i}}(x_i, \hat x_i) + {\rho _{i,{\mathop{\rm int}} }}\left| {{w_i} - {\hat w_i}} \right|^p + {\rho _{i,\rm ext}}|\hat u_i|^{q})
						\\\le \sum\limits_{i = 1}^\infty  {\mu _i}( - {\lambda_i}{V_{i, s_i}}(x_i, \hat x_i) + \sum\limits_{j \in {I_{i,s_i}^{\mathrm{in}}}} {{\rho _{i,{\mathop{\rm int}} }}{\bar N_i}\left| {{w_{ij}} - {\hat w_{ij}}} \right|^p}  + {\rho _{i,\rm ext}}|\hat u_i|^{q}) \\
						\le \sum\limits_{i = 1}^\infty  {\mu_{i}}( - {\lambda_i}{V_{i, s_i}}(x_i, \hat x_i) + \sum\limits_{j \in {I_{i,s_i}^{\mathrm{in}}}} {{\rho _{i,{\mathop{\rm int}} }}{\bar N_i}\left| {{h_{j, s_j}}({x_j}) - {{\hat h}_{j, s_j}}({{\hat x}_j})} \right|^p} + {\rho _{i,\rm ext}}|\hat u_i|^{q}) \\
						\le \sum\limits_{i = 1}^\infty  {\mu_{i}}( - {\lambda_i}{V_{i, s_i}}({x_i},{{\hat x}_i}) + \sum\limits_{j \in {I_{i,s_i}^{\mathrm{in}}}} {{\rho _{i,{\mathop{\rm int}} }}{\bar N_i}\frac{1}{\alpha _j}{V_{j, s_j}}({x_j},{{\hat x}_j})} + {\rho _{i,\rm ext}}|\hat u_i|^{q}) \\
						\mathop  \leq \limits^{(\ref{gamma})} \sum\limits_{i = 1}^\infty  {\mu_{i}}\left( - {\lambda_i}{V_{i, s_i}}({x_i},{{\hat x}_i}) + \sum\limits_{j \in {I_{i,s_i}^{\mathrm{in}}}} {{\gamma_{ij}}{V_{j, s_j}}({x_j},{{\hat x}_j}})  + {\rho _{i,\rm ext}}|\hat u_i|^{q}\right).
					\end{array}
				\end{align}
				\rule{\textwidth}{0.4pt}
			\end{figure*}
			
			Letting $V_{s_{vec}}(x,\hat x) := \left(V_{i,s_i}(x_i, \hat x_i)\right)_{i\in\N}$ and using~\eqref{eq_nablaviest} and~\eqref{eq_muest2}, we have that	
			\begin{align*}
				&V_{s'}(f_{s}(x, u),\hat f_{s}(\hat x, \hat u))- V_{s}(x,\hat x) \nonumber\\&\leq  \Bigl[ \mu\trn(-\Lambda + \Gamma)V_{s_{vec}}(x,\hat x) + \sum_{i=1}^{\infty} {\mu_{i}}{\rho _{i,\rm ext}}|\hat u_i|^{q}\Bigr]\nonumber\\&
				\leq -\lambda_{\infty} V_{s}(x,\hat x) + {\rho _{\rm ext}}(|\hat u|_q), 
			\end{align*}
			where ${\rho _{\rm ext}}(t)= \overline{\mu}\;{{\overline\rho _{\rm{ext}}}}t^q$ for all $t\geq 0$.
		\end{proof}
		\begin{remark}
			The significance of Assumptions~\ref{ass_external_gains} and~\ref{SIM_def} in Theorem~\ref{MTC} has been discussed in~\cite{kawan2019lyapunov}.
			Specifically, the small-gain condition $r(\Psi) < 1$ is tight and cannot be relaxed.
			In view of Gelfand's formula, the spectral small-gain condition is equivalent to the existence of $k\in\N$ such that $\|\Psi^k\|<1$.
For networks with some special structure, e.g. (quasi) spatially invariant systems, one can easily check Assumption~\ref{ass:spectral-radius} with the use of Gelfand's formula,  see Section~\ref{sec:Example} for more details.
		\end{remark}
		\subsection{From infinite to finite networks}
		The main purpose of dealing with infinite networks is to develop scale-free tools for the analysis and design of finite, but arbitrarily large networks.
		In this section we truncate the infinite network $\Sigma$ and keep only the first $n$ subsystems of the network.
		Roughly speaking, we show that if conditions required by Theorem~\ref{MTC} hold, then for any truncation of infinite network $\Sigma$ and accordingly that of $\hat\Sigma$, the same conclusion as in Theorem~\ref{MTC} is obtained for truncated networks.
		
		Consider the first $n\in\N$ subsystems of $\Sigma$ and denote the truncated system by $\Sigma^{<n>}$=$(\X^{<n>}, \X^{l}, \U^{<n>}, \UC^{<n>}, \Y^{<n>}, h^{<n>}_{s^{<n>}}, f^{<n>}_{s^{<n>}}, S^{<n>})$ whose dynamics is described by
		\begin{align}
			\Sigma^{<n>}\!\!:\!\!\!\!\!\!\quad \left\{\begin{array}{l}
				\!\!\!\!\mathbf x^{<n>}(k+1)=\!\!{f^{<n>}_{\sigma^{<n>}(k)}}({\mathbf x^{<n>}(k)},\tilde{\mathbf{x}}(k),{\mathbf u^{<n>}(k)}),\\
				\!\!\!\!\mathbf y^{<n>}(k)= h^{<n>}_{\sigma^{<n>}(k)}(\mathbf x^{<n>}(k)),
			\end{array} \right. 
		\end{align}
		where $\mathbf x^{<n>}(k) \!\!=\!\! (\mathbf x_i(k))_{1\leq i\leq n}$ are elements of $\X^{<n>}\subseteq\R^{N}$, $N:=\sum_{i=1}^{n}n_i$, $\mathbf u^{<n>}(k) \!\!=\!\! (\mathbf u_i(k))_{1\leq i\leq n}$ are elements of $\U^{<n>}\subseteq\R^{M}$ and $M:=\sum_{i=1}^{n}m_i$.
		Moreover, we denote by $I_{\sigma^{<n>}(k)}^{in^{<n>}}=\bigcup\limits_{i = 1}^n {I_{i,{\sigma _i}(k)}^{in}\backslash \left\{ {1, \ldots ,n} \right\}}$, the finite set of neighbors of the first $n$ subsystems.
		Then, $\tilde{\mathbf{x}}(k)=(\mathbf x_j(k))_{j\in I_{\sigma^{<n>}(k)}^{in^{<n>}}}\in\X^{l}\subseteq\R^L$, $L:=\sum_{j\in I_{\sigma^{<n>}(k)}^{in^{<n>}}}n_j$, is considered as the additional input vector.
		Note that we do not neglect subsystems $\Sigma_i$, $i>n$, instead we consider them as additional external inputs $\tilde{\mathbf{x}}(k)$ to the network $\Sigma^{<n>}$.
		Clearly, the case in which subsystems $\Sigma_i$, $i>n$, are entirely removed from the network is covered by our setting by taking $\tilde x \equiv 0$.
		We denote the set of input functions of the truncated network as $\mathcal{U}^{<n>}$ and the output maps are viewed as $ h^{<n>}_{s^{<n>}}:\mathbb{X}^{<n>}\to\mathbb{Y}^{<n>}$ with $S^{<n>}=\prod_{1\leq i\leq n}S_i$.  Moreover, functions ${f^{<n>}_{s^{<n>}}}:\mathbb{X}^{<n>}\times\mathbb{X}^l\times\mathbb{U}^{<n>}\to\mathbb{X}^{<n>}$ are defined accordingly.

		In the following, we construct the compositional construction of abstractions for the network $\Sigma^{<n>}$ under the assumption of Theorem \ref{MTC}.
		\begin{theorem}\label{MTC2}
			Consider the truncated networks $\Sigma^{<n>}=(\mathbb{X}^{<n>}, \X^{l}, \mathbb{U}^{<n>}, \mathcal{U}^{<n>}, \mathbb{Y}^{<n>}, h^{<n>}_{s^{<n>}}, f^{<n>}_{s^{<n>}}, S^{<n>})$ and $\hat \Sigma^{<n>}=(\mathbb{\hat X}^{<n>}, \hat\X^{l}, \mathbb{\hat U}^{<n>}, \mathcal{\hat U}^{<n>}, \mathbb{\hat Y}^{<n>}, \hat h^{<n>}_{s^{<n>}},$ $\hat f^{<n>}_{s^{<n>}}, S^{<n>})$.
			Let $p,q \in [1,\infty)$ be given.
			Consider local simulation functions $V_{i,s_i}:\X_i\times \mathbb{\hat X}_i \rightarrow \R_+, s_i\in S_i$, and suppose that Assumptions~\ref{ass_external_gains},~\ref{SIM_def} and \ref{ass:spectral-radius} hold.
Assume that there exists a vector $\mu = (\mu_{i})_{i\in\N}\in \ell^{\infty}$, $\underline{\mu} \leq \mu_i \leq \overline{\mu}$, with some constants $\ul\mu,\ol\mu>0$ satisfying~\eqref{eq_muest2}.
			Then, the family of functions $V_{s^{<n>}}:\mathbb{X}^{<n>}\times \mathbb{\hat X}^{<n>} \!\!\!\!\rightarrow \R_+$, $s^{<n>}\in S^{<n>}$, where
			\begin{equation*}\label{eq:sim-function-construction2}
				V_{s^{<n>}}(x^{<n>}, \hat x^{<n>})\! = \!\!\sum_{i=1}^{n} \mu_{i} V_{i,s_i}(x_i, \hat x_i), %
			\end{equation*}
			are simulation functions from $\hat \Sigma^{<n>}$ to $\Sigma^{<n>}$ with $b = p , \alpha = \ul\mu	\ul\alpha$ as in~\eqref{sim2} and satisfy the following 
			\begin{align}\label{sim_trunc}
					\begin{array}{*{20}{l}}
						V_{s^{<n>'}}(f_{s^{<n>}}^{<n>}(x^{<n>},\tilde x, u^{<n>}),\hat f_{s^{<n>}}^{<n>}(\hat x^{<n>},\hat{\tilde x}, \hat u^{<n>}))\\-V_{s^{<n>}}(x^{<n>},\hat x^{<n>})\\\leq -\lambda V_{s^{<n>}}(x^{<n>},\hat x^{<n>})+{\rho _{\rm ext}}(|\hat u^{<n>}|_q)+{\rho _{\rm ext}}(|\hat{\tilde x}^{<n>}|_q),
					\end{array}
				\end{align}
				for all $s^{{<n>}'}, s^{<n>}\in S^{<n>}$, where $\lambda = \lambda_\infty$ and $\rho_{\mathrm{ext}} : t \mapsto \overline{\mu}\; \ol\rho _{\rm{ext}}t^q$.
			\end{theorem}
			\begin{proof}
				By following similar arguments as in Theorem \ref{MTC}, one can obtain
				\begin{align*}
					\sum_{i=1}^{n} \mu_i V_{i, s_i}(x_i, \hat x_i) &\geq \sum_{i=1}^{n} \mu_i \alpha_i|h_{i, s_{i}}(x_i) - \hat h_{i, s_{i}}(\hat x_i)|^p\\
					& \geq \ul\mu	\ul\alpha \sum_{i=1}^{n} |h_{ii, s_{i}}(x_i) - \hat h_{ii, s_{i}}(\hat x_i)|^p \\
					& \geq \ul\mu	\ul\alpha |h^{<n>}_{s^{<n>}}(x^{<n>}) - \hat h^{<n>}_{s^{<n>}}(\hat x^{<n>})|_p^p  .
				\end{align*}
				Moreover, by letting $V_{s^{<n>}_{vec}}(x^{<n>},\hat x^{<n>}) := \left(V_{i,s_i}(x_i, \hat x_i)\right)_{1\leq i\leq n}$, using the chain of inequalities in \eqref{eq_nablaviest} for all $s'_i, s_i\in S_i$, $s_j\in S_j$, $s^{{<n>}'}, s^{<n>}\in S^{<n>}$, $1\leq i\leq n$, and~\eqref{eq_muest2} , we have	
				\begin{align*}
					&V_{s^{{<n>}'}}(f^{<n>}_{{s^{<n>}}}(x^{<n>},\tilde x, u^{<n>}),\hat f^{<n>}_{s^{<n>}}(\hat x^{<n>}, \hat {\tilde x}, \hat u^{<n>}))\nonumber\\&- V_{{s^{<n>}}}(x^{<n>},\hat x^{<n>})\leq \nonumber\\&  \Bigl\{\Bigl[\mu\trn(-\Lambda + \Gamma)\Bigr]_{1\leq i\leq n} V_{s^{<n>}_{vec}}(x,\hat x) + \sum_{i=1}^{n} {\mu_{i}}{\rho _{i,\rm ext}}|\hat{\tilde x}_i|^{q}\nonumber\\&+\sum_{i=1}^{n} {\mu_{i}}{\rho _{i,\rm ext}}|\hat u_i|^{q}\Bigr\}\nonumber\\&
					\leq -\lambda_{\infty} V_{s^{<n>}}(x^{<n>},\hat x^{<n>}) + \overline{\mu}\;{{\overline\rho _{\rm{ext}}}}|\hat{\tilde x}^{<n>}|_q^q\\&\;\;\;\;\;+\overline{\mu}\;{{\overline\rho _{\rm{ext}}}}|\hat u^{<n>}|_q^q.
				\end{align*}
			\end{proof}
			As can be seen from Theorem \ref{MTC2}, the decay rate $\lambda_{\infty}$ as well as the gain function due to external input $\hat u$ are preserved under truncation.
				Thus, the indices of the proposed compositional method are \emph{independent} of the network size.
			\section{Construction of Abstractions for Linear Systems}\label{linear}
			In this section, we \emph{explicitly} construct local abstractions and corresponding simulation functions for linear switched subsystems.
			
			We make the following assumption on the simulation functions, which is an incremental version of a similar assumption used to achieve the input-to-state stability of switched systems under constrained switching conditions~\cite{vu2007input}.
			\begin{assumption}\label{cons_switch}
				There exist uniformly bounded constants $\tau_i\geq 1$, $i\in \N$,  such that for all $x_i \in \X_i$, all $\hat x_i \in \mathbb{\hat X}_i$ and every $s_i,s'_i\in S_i$
				\begin{equation*}
					{V_{i, s_i}}({x_i},{{\hat x}_i})\leq \tau_i {V_{i, s'_i}}({x_i},{{\hat x}_i}).
				\end{equation*}
			\end{assumption}	
			Consider the following class of linear switched subsystems
			\begin{equation}\label{eq_ith_subsystem_lin}
			{\Sigma _i}: \left\{ \begin{array}{l}
			\mathbf x_i(k+1)  = A_{i, \sigma_i(k)}{\mathbf x}_i(k)+D_{i, \sigma_i(k)}{\mathbf w_i(k)}\\\;\;\;\;\;\;\;\;\;\;\;\;\;\;\;\;\;\;\;+B_{i, \sigma_i(k)}{\mathbf u}_i(k) ,\\
			{\mathbf y_i(k)} = C_{i, \sigma_i(k)}{\mathbf x_i(k)},
			\end{array} \right.%
			\end{equation}
			where $\sigma_i\in{\cal{S}}_i$, $A_{i, \sigma_i(k)}\in \R^{n_i\times n_i}$, $B_{i, \sigma_i(k)}\in \R^{n_i\times m_i}$, $C_{i, \sigma_i(k)}\in \R^{q_i\times n_i}$ and $D_{i, \sigma_i(k)}\in \R^{n_i\times p_i}$, for $i\in\N$.
			
			Choose $ \X = \ell^2(\N,(n_i))$ and $\mathbb{U} = \ell^2(\N,(m_i))$ for the overall infinite network. 
			By slight abuse of notation, we use the tuple $\Sigma_i=({A_{i, s_i}}, {B_{i, s_i}},$ ${C_{i, s_i}}, {D_{i, s_i}})$ to refer to switched subsystem with transition and output functions of the form \eqref{eq_ith_subsystem_lin} with the specified matrices dimensions.
			
			Assume that there exist a family of matrices $K_{i, s_i}$, positive definite matrices $M_{i, s_i}$, real numbers $\epsilon_i>0$, and $0<\kappa_i<1$ such that the following matrix inequalities hold for all $s_i\in S_i$, $i\in\N$
			\begin{subequations}\label{cond22}
				\begin{align}
					&C_{i, s_i}^\top{C_{i, s_i}} \preceq {M_{i, s_i}},\label{cond222}\\&
					(1+\frac{1}{\epsilon_i}+\epsilon_i){({A_{i, s_i}} + {B_{i, s_i}}{K_{i, s_i}})^\top}{M_{i, s_i}}({A_{i, s_i}} + {B_{i, s_i}}{K_{i, s_i}})\nonumber\\& \preceq  \kappa_i M_{i, s_i}.\label{cond2222}
				\end{align}
			\end{subequations}	
			\begin{remark}
				Given $\kappa_i$ and $\epsilon_i$, inequality \eqref{cond2222} is not jointly convex on the decision variables $M_{i,s_i}$ and $K_{i,s_i}$. Then, this inequality is not amenable to existing semidefinite tools for linear matrix inequalities (LMI). By using the Schur complement lemma, \eqref{cond2222} could be transformed to the following LMI over decision variables $Q_{i,s_i}$ and $Z_{i,s_i}$:
				\begin{align*}
					&\left[ {\begin{array}{*{20}{c}}
							{ - {\kappa _i}{Q_{i,{s_i}}}}&{{Q_{i,{s_i}}}{A}_{i,{s_i}}^\top + {Z}_{i,{s_i}}^\top{B}_{i,{s_i}}^\top}\\
							{{A_{i,{s_i}}}{Q_{i,{s_i}}}}+{B}_{i,{s_i}}{Z}_{i,{s_i}}&{ - (1+\frac{1}{\epsilon_i}+\epsilon_i){Q_{i,{s_i}}}}
						\end{array}} \right] \preceq 0,\\&\;\;{Q_{i,{s_i}}} \succ 0,
					\end{align*}
					where ${Q_{i,{s_i}}}=M_{i,{s_i}}^{-1}$ and ${Z_{i,{s_i}}}={K_{i,{s_i}}}{Q_{i,{s_i}}}$, $i\in\N$.
				\end{remark}
				Consider the simulation function candidates $V_{i,s_i}:\mathbb{R}^{n_i}\times \mathbb{R}^{\hat n_i} \rightarrow \R_+, s_i\in S_i$, $i\in\N$, as
				\begin{align}\label{sim_func}
					{V_{i, s_i}}({x_i},{{\hat x}_i}) = {({x_i} - {P_{i, s_i}}{{\hat x}_i})^\top}{M_{i, s_i}}({x_i} - {P_{i, s_i}}{{\hat x}_i}).
				\end{align}
				The control inputs of the concrete subsystems are given by the interface functions $\nu_i$ as follows.
				\begin{align}\label{interface}
					u_i&=\nu_i (x_i,{\hat x}_i,{\hat u}_i,{\hat w}_i,s_i) \\\notag
					= & K_{i, s_i}({x_i} - P_{i, s_i} {\hat x}_i) + Q_{i, s_i} {\hat x}_i+ R_{i, s_i}{\hat u}_i + T_{i, s_i} {\hat w}_i ,
				\end{align}
				where $P_{i, s_i}$, $i\in\N$, are some matrices of appropriate dimensions.
				Assume that the following inequalities hold for some matrices of appropriate dimensions $Q_{i, s_i}$, $T_{i, s_i}$.
				\begin{subequations}\label{cond2}
					\begin{align}
						&A_{i, s_i} P_{i, s_i} =  P_{i, s_i} {\hat A}_{i, s_i} - B_{i, s_i} Q_{i, s_i} , \label{cond2_1}\\
						&{D_{i, s_i}} = {P_{i, s_i}}{{\hat D}_{i, s_i}} - {B_{i, s_i}}{T_{i, s_i}},\label{cond2_2}\\
						&{C_{i, s_i}}{P_{i, s_i}} = {{\hat C}_{i, s_i}}.\label{cond2_3}
					\end{align}
				\end{subequations}
				Next theorem shows that functions $V_{i, s_i}$, $s_i\in S_i$, defined in \eqref{sim_func}, are simulation functions from $\hat\Sigma_i$ to $\Sigma_i$.
				\begin{theorem}
					Consider systems $\Sigma_i=({A_{i, s_i}}, {B_{i, s_i}},$ ${C_{i, s_i}}, {D_{i, s_i}})$ and $\hat \Sigma_i=({\hat A_{i, s_i}}, {\hat B_{i, s_i}}, {\hat C_{i, s_i}}, {\hat D_{i, s_i}})$ for $i\in\N$. Suppose that for all $s_i\in S_i$, there exist appropriate matrices $M_{i, s_i}$, $P_{i, s_i}$, $K_{i, s_i}$, $Q_{i, s_i}$ and $T_{i, s_i}$ satisfying \eqref{cond22} and \eqref{cond2}. Moreover, assume that $\tau_i\kappa_i<1$. Then, functions in \eqref{sim_func}  are simulation functions from $\hat \Sigma_i$ to $\Sigma_i$ with concrete inputs given by \eqref{interface}.
				\end{theorem}
				\begin{proof}
					According to \eqref{cond2_3}, we have 
					\begin{align*}
						&| C_{i, s_i}x_i \!- \!\hat{C}_{i, s_i}\hat{x}_i| =\\
						&\big((x_i\!-\!P_{i, s_i}\hat C_{i, s_i})^\top \!C_{i, s_i}^\top \!C_{i, s_i}(x_i\!-\!P_{i, s_i}\hat C_{i, s_i})\big)^\frac{1}{2}.
					\end{align*}
					Using \eqref{cond222}, it is clear that $| {C_{i, s_i}x_i - \hat C_{i, s_i}\hat x_i}|^2 \le V_{i, s_i}({x_i},{{\hat x}_i})$ holds for all $x_i \in \X_i$, $\hat x_i \in \mathbb{\hat X}_i$.
					Then, \eqref{sim_eq_viest} is satisfied with $\alpha_i=1, i\in\N$, $p=2$. 
					
					Now, we proceed to show that \eqref{sim_eq_nablaviest} is satisfied, too. 
					
					Using Assumption \ref{cons_switch}, one gets the following inequality for all switchings $s'_i,s_i\in S_i$
					\begin{align}\label{ass}
						{V_{i, s'_i}}\left(f_{i,s_i}(x_i, w_i, u_{i}),\hat f_{i,s_i}(\hat x_i, \hat w_i, \hat u_i)\right)\leq \nonumber\\\tau_i 	{V_{i, s_i}}\left(f_{i,s_i}(x_i, w_i, u_{i}),\hat f_{i,s_i}(\hat x_i, \hat w_i, \hat u_i)\right).
					\end{align}
					Using the system dynamics \eqref{eq_ith_subsystem_lin} and the candidate simulation function in \eqref{sim_func}, the inequality \eqref{ass} can be written as
					\begin{align}\label{2nd}
						&V_{i,s_i'}\left(f_{i,s_i}(x_i, w_i, u_{i}),\hat f_{i,s_i}(\hat x_i, \hat w_i, \hat u_i)\right)
						\le \nonumber\\&\tau_i[A_{i, s_i}x_i + B_{i, s_i}u_{i} + D_{i, s_i}w_i\nonumber\\&- P_{i, s_i}(\hat A_{i, s_i}\hat x_i + \hat B_{i, s_i}\hat u_i + \hat D_{i, s_i}\hat w_i)]^\top M_{i, s_i}\nonumber\\&\times[A_{i, s_i}x_i + B_{i, s_i}u_{i} + D_{i, s_i}w_i\nonumber\\&- P_{i, s_i}(\hat A_{i, s_i}\hat x_i + \hat B_{i, s_i}\hat u_i + \hat D_{i, s_i}\hat w_i)].
					\end{align}	
				Substituting $u_i$ from \eqref{interface} and employing~\eqref{cond2_1} to~\eqref{cond2_2} yield
					\begin{align*}
						&V_{i,s_i'}\left(f_{i,s_i}(x_i, w_i, u_{i}),\hat f_{i,s_i}(\hat x_i, \hat w_i, \hat u_i)\right)
						\le \nonumber\\&\tau_i
						[(A_{i, s_i} + B_{i, s_i}K_{i, s_i})(x_i - P_{i, s_i}\hat x_i) + D_{i, s_i}(w_i - \hat w_i)\nonumber\\&+ (B_{i, s_i}R_{i, s_i} - P_{i, s_i}\hat B_{i, s_i})\hat u_i]^\top{M_{i, s_i}}\nonumber\\&\times[(A_{i, s_i} + B_{i, s_i}K_{i, s_i})(x_i - P_{i, s_i}\hat x_i) + D_{i, s_i}(w_i - \hat w_i)\nonumber\\&+ (B_{i, s_i}R_{i, s_i} - P_{i, s_i}\hat B_{i, s_i})\hat u_i].
					\end{align*}	
					Applying Young's inequality as $ab\leq \frac{\epsilon}{2}a^2+\frac{1}{2\epsilon}b^2$ for any $a,b\geq0$ and any $\epsilon>0$, we have \eqref{3rd cond}.
					
					\begin{figure*}[ht]
						\rule{\textwidth}{0.4pt}
						\begin{align}\label{3rd cond}
							&V_{i,s_i'}\left(f_{i,s_i}(x_i, w_i, u_{i}),\hat f_{i,s_i}(\hat x_i, \hat w_i, \hat u_i)\right)
							\nonumber\\&\le\tau_i\Bigg({({x_i} - {P_{i, s_i}}{{\hat x}_i})^\top}[{({A_{i, s_i}} + {B_{i, s_i}}{K_{i, s_i}})^\top} M_{i, s_i}({A_{i, s_i}} + {B_{i, s_i}}{K_{i, s_i}}) ]{({x_i} - {P_{i, s_i}}{{\hat x}_i})}\nonumber\\&\;\;\;
							+ [2{({x_i} - {P_{i, s_i}}{{\hat x}_i})^\top}{({A_{i, s_i}} + {B_{i, s_i}}{K_{i, s_i}})^\top}]{M_{i, s_i}}[{D_{i, s_i}}(w_i-\hat w_i)]\nonumber\\
							&\;\;\;+[2{({x_i} - {P_{i, s_i}}{{\hat x}_i})^\top}{({A_{i, s_i}} + {B_{i, s_i}}{K_{i, s_i}})^\top}]{M_{i, s_i}}[(B_{i, s_i}R_{i, s_i}-P_{i, s_i}\hat B_{i, s_i})\hat u_i]\nonumber\\
							&\;\;\;+[2(w_i-\hat w_i)^\top {D_{i, s_i}}^\top]{M_{i, s_i}}[(B_{i, s_i}R_{i, s_i}-P_{i, s_i}\hat B_{i, s_i})\hat u_i]\nonumber\\&\;\;\;+|\sqrt{ M_{i, s_i}}D_{i, s_i}(w_i-\hat w_i)|^2+|\sqrt{ M_{i, s_i}}(B_{i, s_i}R_{i, s_i}-P_{i, s_i}\hat B_{i, s_i})\hat u_i|^2\Bigg) \nonumber\\
							&\leq\tau_i\Bigg({({x_i} - {P_{i, s_i}}{{\hat x}_i})^\top}[{({A_{i, s_i}} + {B_{i, s_i}}{K_{i, s_i}})^\top} M_{i, s_i}({A_{i, s_i}} + {B_{i, s_i}}{K_{i, s_i}})]{({x_i} - {P_{i, s_i}}{{\hat x}_i})}\nonumber\\
							&\;\;\;+\epsilon_i{({x_i} - {P_{i, s_i}}{{\hat x}_i})^\top}[{({A_{i, s_i}} + {B_{i, s_i}}{K_{i, s_i}})^\top} M_{i, s_i}({A_{i, s_i}} + {B_{i, s_i}}{K_{i, s_i}})]{({x_i} - {P_{i, s_i}}{{\hat x}_i})}+\frac{1}{\epsilon_i}|\sqrt{ M_{i, s_i}}D_{i, s_i}(w_i-\hat w_i)|^2\nonumber\\&\;\;\;
							+\frac{1}{\epsilon_i}{({x_i} - {P_{i, s_i}}{{\hat x}_i})^\top}[{({A_{i, s_i}} + {B_{i, s_i}}{K_{i, s_i}})^\top} M_{i, s_i}({A_{i, s_i}} + {B_{i, s_i}}{K_{i, s_i}})]{({x_i} - {P_{i, s_i}}{{\hat x}_i})}\nonumber\\&\;\;\;
							+\epsilon_i|\sqrt{ M_{i, s_i}}(B_{i, s_i}R_{i, s_i}-P_{i, s_i}\hat B_{i, s_i})\hat u_i|^2+\frac{1}{\epsilon_i}|\sqrt{ M_{i, s_i}}(B_{i, s_i}R_{i, s_i}-P_{i, s_i}\hat B_{i, s_i})\hat u_i|^2+\epsilon_i|\sqrt{ M_{i, s_i}}D_{i, s_i}(w_i-\hat w_i)|^2\Bigg).
						\end{align}
						\rule{\textwidth}{0.4pt}
					\end{figure*}

					By employing \eqref{cond2222}, one has
					\begin{align*}
						&V_{i,s_i'}\left(f_{i,s_i}(x_i, w_i, u_i),\hat f_{i,s_i}(\hat x_i, \hat w_i, \hat u_i)\right)\leq\nonumber\\&   \tau_i\Bigg(\kappa_i V_{i,s_i}(x_i,\hat x_i) \nonumber\\&+(1+\frac{1}{\epsilon_i}+\epsilon_i)|\sqrt{ M_{i, s_i}}D_{i, s_i}|^2|w_i-\hat w_i|^2\nonumber\\
						&+(1+\frac{1}{\epsilon_i}+\epsilon_i)\Big|\sqrt{ M_{i, s_i}}(B_{i, s_i}R_{i, s_i}-P_{i,s_i}\hat B_{i, s_i})\Big|^2|\hat u_i|^2\Bigg). 
					\end{align*}
					Since $\tau_i\kappa_i<1$ by assumption, one can define $\hat{\kappa}_i=1-\tau_i\kappa_i$, and rewrite the previous inequality as follows 
					\begin{align*}
						&V_{i,s_i'}\left(f_{i,s_i}(x_i, w_i, u_i),\hat f_{i,s_i}(\hat x_i, \hat w_i, \hat u_i)\right)-V_{i,s_i}(x_i,\hat x_i)\leq \nonumber\\&  -\hat{\kappa}_i V_{i,s_i}(x_i,\hat x_i) \nonumber\\&+\tau_i(1+\frac{1}{\epsilon_i}+\epsilon_i)|\sqrt{ M_{i, s_i}}D_{i, s_i}|^2|w_i-\hat w_i|^2\nonumber\\&
						+\tau_i(1+\frac{1}{\epsilon_i}+\epsilon_i)\Big|\sqrt{ M_{i, s_i}}(B_{i, s_i}R_{i, s_i}-P_{i,s_i}\hat B_{i, s_i})\Big|^2|\hat u_i|^2. 
					\end{align*}
					Thus, \eqref{sim_eq_nablaviest} holds with $p=q=2$, $\lambda_i=\hat{\kappa}_i$, ${\rho _{\rm{i,ext}}}=\tau_i(1+\frac{1}{\epsilon_i}+\epsilon_i)\mathop {\max }\limits_{{s_i}} \{|\sqrt{ M_{i, s_i}}(B_{i, s_i}R_{i, s_i}-P_{i,s_i}\hat B_{i, s_i})|^2 \}$ and ${\rho _{{\mathop{i,\rm int}} }}=\tau_i(1+\frac{1}{\epsilon_i}+\epsilon_i)\mathop{\max }\limits_{{s_i}} \{|\sqrt{ M_{i, s_i}}D_{i, s_i}|^2\}$.
					
					Therefore, the candidate functions in \eqref{sim_func} are simulation functions from $\hat\Sigma_i$ to $\Sigma_i$, for all $ i\in\N$. 
				\end{proof}
				\section{Example}\label{sec:Example}
				To verify the effectiveness of our results, we apply them to a voltage regulation problem in AC islanded microgrids.
				
				Islanded microgrids are self-sufficient small-scale power grids composed of several Distributed Generation Units (DGUs).
				They are designed to operate safely and reliably in the absence of connection to the main grid~\cite{riverso2014plug}. When the microgrids are working in connected mode, voltage and frequency are set by the main grid. However, in the islanded mode, they must be controlled by DGUs. Therefore, their connection should be robust against line faults or variations in the topology of DGUs' connections.
				Treating time-varying communication topologies is beneficial to evaluate the system performance in the presence of the line switches or plug-and-play operations.
				
				We consider a switched AC islanded microgrid network modeled by an interconnection of fourth-order DGUs as underlying subsystems.
				In particular, we consider two \emph{circular} topologies as shown in Figures~\ref{block1} and~\ref{block2} and assume that the network topology switches between these two configurations at certain times.
				Let ${\sigma_i(k)}$ be the switching signal which takes values in the set $\{1,2\}$, where $\sigma_i(k) = 1$ corresponds to the topology shown in Figure~\ref{block1} and $\sigma_i(k) = 2$ pertains to that in Figure~\ref{block2}.
				
				The discrete-time dynamics of each DGU in the microgrid with sampling time $t_s$  is described by~\eqref{power_subsystem}, adapted from~\cite{riverso2014plug}. 
				\begin{figure*}[ht]
					\rule{\textwidth}{0.4pt}
					\begin{align}\label{power_subsystem}
						\Sigma _i:\!\left\{ \begin{array}{l}
							\underbrace{\left[ \begin{array}{*{20}{c}}
									\mathbf{V}_{i,d} (k + 1) \\
									\mathbf{V}_{i,q} (k + 1) \\
									\mathbf{I}_{ti,d} (k + 1) \\
									\mathbf{I}_{ti,q} (k + 1)\\
									\nu_{i,d}(k+1)\\
									\nu_{i,q}(k+1)
								\end{array} \right]}_ {=:\mathbf x_i (k+1)}\!\!=\!\! \underbrace{\left[ \begin{array}{*{20}{c}}
								-\frac{ t_s}{C_{ti}}(\sum\nolimits_{j \in {I_{i,{\sigma_i(k)}}^{\mathrm{in}}}} {\frac{{{R_{ij}}}}{{Z_{ij}^2}}} )+1& t_s\omega_0-\frac{ t_s}{C_{ti}}(\sum\nolimits_{j \in {I_{i,{\sigma_i(k)}}^{\mathrm{in}}}} {\frac{{{X_{ij}}}}{{Z_{ij}^2}}} )&\frac{ t_sk_i}{C_{ti}}&0&0&0\\
								- t_s\omega_0+\frac{ t_s}{C_{ti}}(\sum\nolimits_{j \in {I_{i,{\sigma_i(k)}}^{\mathrm{in}}}} {\frac{{{X_{ij}}}}{{Z_{ij}^2}}} )&	-\frac{ t_s}{C_{ti}}(\sum\nolimits_{j \in {I_{i,{\sigma_i(k)}}^{\mathrm{in}}}} {\frac{{{R_{ij}}}}{{Z_{ij}^2}}} )+1&0&\frac{ t_sk_i}{C_{ti}}&0&0\\
								-\frac{ t_sk_i}{L_{ti}}&0&-\frac{ t_sR_{ti}}{L_{ti}}+1& t_s\omega_0&0&0\\
								0&-\frac{ t_sk_i}{L_{ti}}&- t_s\omega_0&-\frac{ t_sR_{ti}}{L_{ti}}+1&0&0\\-t_s&0&0&0&1&0\\0&-t_s&0&0&0&1\\
							\end{array} \right]}_{=:A_{i,{\sigma_i (k)}}} \\
						\;\;\;\;\;\;\;\;\;\;\;\;\;\;\;\;\;\;\;\;\;\;\;\;\;\;
						\times\underbrace{\left[ \begin{array}{*{20}{c}}
								\mathbf{V}_{i,d}(k)\\
								\mathbf{V}_{i,q}(k)\\
								\mathbf{I}_{ti,d} (k) \\
								\mathbf{I}_{ti,q} (k)\\
								\nu_{i,d}(k)\\
								\nu_{i,q}(k)
							\end{array} \right]}_{=:\mathbf x_i (k)}+ \underbrace{\frac{ t_s}{C_{ti}} \sum\nolimits_{j \in {I_{i,{\sigma_i(k)}}^{\mathrm{in}}}}\left[ {\begin{array}{*{20}{c}}
								{\frac{{{R_{ij}}}}{{Z_{ij}^2}}}&  {\frac{{{X_{ij}}}}{{Z_{ij}^2}}}\\
								-{\frac{{{X_{ij}}}}{{Z_{ij}^2}}}& {\frac{{{R_{ij}}}}{{Z_{ij}^2}}}\\
								0&0\\0&0\\0&0\\0&0
							\end{array}} \right]}_{=:D_{i,{\sigma_i(k)}}} \,\,\underbrace{\left[ \begin{array}{*{20}{c}}
							\mathbf{V}_{j,d}(k)\\
							\mathbf{V}_{j,q}(k)
						\end{array} \right]}_{=: \mathbf w_i (k)}\\
					\;\;\;\;\;\;\;\;\;\;\;\;\;\;\;\;\;\;\;\;\;\;\;\;\;\;
					+\underbrace{\left[ {\begin{array}{*{20}{c}}
								-{\frac{1}{C_{ti}}}&0&0&0\\
								0& -{\frac{1}{C_{ti}}}&0&0\\
								0&0&0&0\\
								0&0&0&0\\
								0&0&t_s&0\\0&0&0&t_s\\
							\end{array}} \right]}_{=:H_{i}} \,\,\underbrace{\left[ \begin{array}{*{20}{c}}
							\mathbf{I}_{Li,d}\\
							\mathbf{I}_{Li,q}\\
							\mathbf{y}_{i,d,ref}\\
							\mathbf{y}_{i,q,ref}\\
						\end{array} \right]}_{=: \mathbf d_i}
					+ \underbrace{\left[ {\begin{array}{*{20}{c}}
								0&0\\
								0&0\\
								\frac{ t_s}{L_{ti}}&0\\
								0&\frac{t_s}{L_{ti}}\\
								0&0\\
								0&0
							\end{array}} \right]}_{=:B_{i}} {\mathbf u}_{i} (k),\\
						{\mathbf y_{i}}(k) = \underbrace{\left[ \begin{array}{*{20}{c}}
								1&0&0&0&0&0\\
								0&1&0&0&0&0
							\end{array} \right]}_{=:C_{i}} \left[ \begin{array}{*{20}{c}}
							\mathbf{V}_{i,d}(k)\\
							\mathbf{V}_{i,q}(k)\\
							\mathbf{I}_{ti,d} (k) \\
							\mathbf{I}_{ti,q} (k)\\
							\nu_{i,d}(k)\\
							\nu_{i,q}(k)
						\end{array} \right].
					\end{array} \right.
				\end{align}
				\rule{\textwidth}{0.4pt}
			\end{figure*}
			In \eqref{power_subsystem}, $\mathbf{V}_{i,d}$ (resp. $\mathbf{V}_{i,q}$) are the $d$ (resp. $q$) components of the load voltage. Similarly, $\mathbf{I}_{ti,d}$ (resp. $\mathbf{I}_{ti,q}$) denote the $d$ (resp. $q$) components of the current of DGU $\Sigma_i$. 
			In addition, the integrators $\nu_{i,d}$, $\nu_{i,q}$ are added for disturbance rejection reasons~\cite{riverso2014plug}.
			The control inputs (the voltage of corresponding voltage source converter (VSC)) and outputs are denoted by $u_i(k)=\left[ {\mathbf{V}_{ti,d}(k),\mathbf{V}_{ti,q}(k)} \right]^\top$ and $\mathbf{y}_i(k)=\left[ {\mathbf{V}_{i,d}(k),\mathbf{V}_{i,q}(k)} \right]^\top$, respectively.
			Furthermore, $D_{i,{\sigma_i(k)}}\mathbf w_i (k)$ models the coupling of DGU $\Sigma_i$ with its neighbors $\Sigma_j$, $j\in I_{i,{\sigma_i(k)}}^{\mathrm{in}}$, corresponding to each switching mode. In addition, $H_{i}\mathbf d_i$ represents the collection of load currents $I_{Li,d}$ and $I_{Li,q}$ which are considered as constant exogenous inputs acting as a disturbance and tracking references $\mathbf{y}_{i,ref}=\left[ {\mathbf{y}_{i,d,ref},\mathbf{y}_{i,q,ref}} \right]^\top$.
			
The parameters $R_{ti}$, $L_{ti}$ are the resistance and inductance, respectively, corresponding to DGU $\Sigma_i$ and $R_{ij}$, $L_{ij}$ are those of the line between DGU $\Sigma_i$ and DGU $\Sigma_j$ which are connected through a three-phase line.
In addition, $X_{ij}=\omega_0L_{ij}$ and $Z_{ij}=|R_{ij}+jX_{ij}|$ with the rotation speed $\omega_0$. Moreover, $k_i$ is the transformer ratio which connects DGU $\Sigma_i$ to the remainder of the network.
			The other transformer parameters are included in $R_{ti}$ and $L_{ti}$. A shunt capacitance $C_{ti}$ is used for attenuating the impact of high-frequency harmonics of the load voltage. 
			
			The interconnection structure switches between two \emph{circular} topologies shown in Figures~\ref{block1} and~\ref{block2}.
			In these topologies, each subsystem $\Sigma_i$ is fed by subsystems $\Sigma_{i-1}$ for $\sigma_i(k) = 1$ ($I_{i,1}^{\mathrm{in}}=\{i-1\}$) and $\Sigma_{i+1}$ for $\sigma_i(k) = 2$ ($I_{i,2}^{\mathrm{in}}=\{i+1\}$), respectively.
			\begin{figure}
				\centering  
				\includegraphics[width=7.5cm]{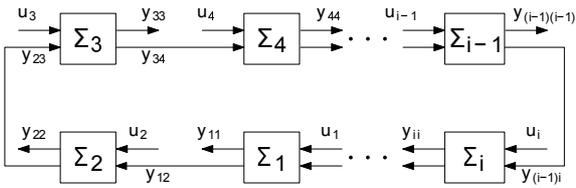}
				\caption{The interconnected system $\Sigma$ for $s_i=1$.}\label{block1}
			\end{figure}
			
			\begin{figure}
				\centering  
				\includegraphics[width=7.5cm]{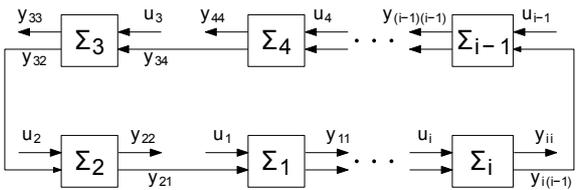}
				\caption{The interconnected system $\Sigma$ for $s_i=2$.}\label{block2}
			\end{figure}

			We denote $\Sigma$ as the augmented infinite network consisting of infinite subsystems $\Sigma_i$. To construct an overall abstraction for $\Sigma$, we construct abstractions of subsystems $\Sigma_i$, $i\in \N$, with dimensions $\hat n_i$ for both $s_i =1,2$.
			Necessary and sufficient conditions on the geometrical properties of the involved matrices ${{P}_{i, s_i}}$, ${{D}_{i, s_i}}$, ${{A}_{i, s_i}}$, ${{B}_{i, s_i}}$ in~\eqref{cond2} are provided in \cite[Sec. 4.3]{rungger2016compositional}, which determine the lowest possible state dimension for $\hat\Sigma_i$, $i\in\N$, as $\hat n_i=3$. 
			
			Now we compute the abstraction matrices satisfying~\eqref{cond2}.
			Considering~\eqref{cond2_1} and~\eqref{cond2_2} and  taking ${{\hat D}_{i, s_i}}= t_s\left[ {\begin{array}{*{20}{c}}
				1&0&1\\
				0&1&0\\
				\end{array}} \right]^\top$ and ${T_{i, s_i}}=0$, we get 
			\begin{align*}
				&P_{i,s_i}=\frac{1}{2C_{ti}} \sum\nolimits_{j \in {I_{i,{s_i}}}} \left[ {\begin{array}{*{20}{c}}
						0&0&0&0&1&1\\
						{\frac{X_{ij}}{{Z_{ij}^2}}}&{\frac{{{R_{ij}}}}{{Z_{ij}^2}}}&0&0&0&0\\
						{\frac{{{R_{ij}}}}{{Z_{ij}^2}}}&{-\frac{X_{ij}}{{Z_{ij}^2}}}&0&0&-1&-1\\
					\end{array}} \right]^\top ,\\&
					Q_{i,s_i} =\frac{t_s k_i}{2C_{ti}} \sum\nolimits_{j \in {I_{i,{s_i}}}}\left[ {\begin{array}{*{20}{c}}
							0&{\frac{-X_{ij}}{{Z_{ij}^2}}}&{\frac{-R_{ij}}{{Z_{ij}^2}}}\\
							0&{\frac{-R_{ij}}{{Z_{ij}^2}}}&{\frac{X_{ij}}{{Z_{ij}^2}}}
						\end{array}} \right], s_i=1,2.
					\end{align*}
					Furthermore, $\hat A_{i,{s_i}}$ is obtained by solving $\hat n_i\times\hat n_i$ equations provided that matrix $\sum\nolimits_{j \in {N_{i,{s_i}}}}\left[ {\begin{array}{*{20}{c}}
						{\frac{{{R_{ij}}}}{{Z_{ij}^2}}}&{\frac{-X_{ij}}{{Z_{ij}^2}}}\\
						{\frac{X_{ij}}{{Z_{ij}^2}}}&{\frac{{{R_{ij}}}}{{Z_{ij}^2}}}
						\end{array}} \right]$ is invertible.  
					
					In addition, $\hat C_{i, s_i}=\frac{1}{2C_{ti}} \sum\nolimits_{j \in {N_{i,{s_i}}}} \left[ {\begin{array}{*{20}{c}}
						0&{\frac{X_{ij}}{{Z_{ij}^2}}}&{\frac{{{R_{ij}}}}{{Z_{ij}^2}}}\\
						0&{\frac{{{R_{ij}}}}{{Z_{ij}^2}}}&{-\frac{X_{ij}}{{Z_{ij}^2}}}\\
						\end{array}} \right]$. 
					By considering the computed matrices $\hat A_{i, s_i}$ and taking $\hat B_{i, s_i}=I_{\hat n_i}$, we choose appropriate matrices $\hat K_{i, s_i}$ for local controllers $\hat u_i=-\hat K_{i, s_i}\hat x_i$, which stabilize abstract subsystems $\hat\Sigma_i$ at the origin.
					
					We also choose $R_{i, s_i}=(B_{i, s_i}^\top M_{i, s_i}B_{i, s_i})^{-1} B_{i, s_i}^\top M_{i, s_i}P_{i,s_i}\hat B_{i, s_i}$ to minimize $\rho_{\rm{i,ext}}$ as suggested in~\cite{girard2009hierarchical}.
					
					We illustrate the \emph{scale-free} property of our approach with respect to the size of network via simulations.
					Following Theorem~\ref{MTC2}, we consider three truncated networks of microgrids shown in Figures~\ref{block1} and~\ref{block2}, respectively, consisting of $10^2$, $10^3$ and $10^4$ subsystems.
					The parameters are set as $R_{ti}=1.5\rm{m\Omega }$, $L_{ti}=300 \rm{\mu H}$, $C_{ti}=460 \rm{\mu F}$, $k_i=1$ for all subsystems $\Sigma_i$.
					Additionally, we choose $R_{ij}=1\rm{m\Omega }$, $L_{ij}=10\rm{mH}$ for all subsystems $\Sigma_i,\Sigma_j$ with $s_i=1$ and $R_{ij}=1.2\rm{m\Omega }$, $L_{ij}=8\rm{mH}$ for all subsystems $\Sigma_i,\Sigma_j$ with $s_i=2$.
					The microgrids frequency and the sampling time are set as $f_0=60\rm{Hz}$ and $t_s=10^{-4}\rm{s}$, respectively.				
					The switchings between $s_i=1$ and $s_i=2$ occur at time steps $k=4n$, $n\in \N$.
We choose $\kappa_i=0.01$ and take matrices $K_{i,s_i}$ such that the eigenvalues of pairs $(A_{i,s_i}, B_{i,s_i})$ in closed loop are $\left[0.3;0.15;0.6;0.2;0.4;0.5 \right]$ for both $s_i =1,2$.
Then, we compute $M_{i,{s_i}}=\left[ {\begin{array}{*{20}{c}}
						12.291 &-0.473&11.082&-1.081&7.809&-1.665\\
						*&23.041&0.535&17.258&1.780&-0.585\\
						*&*&37.993&26.610&-0.607&-0.458\\
						*&*&*&47.840&0.536&-0.358\\
						*&*&*&*&20.913&2.330\\
						*&*&*&*&*&23.559
						\end{array}} \right]$, $s_i=1$, and $M_{i,{s_i}}=\left[ {\begin{array}{*{20}{c}}
						13.013 &-0.801&10.235&-2.132&6.296&-1.819\\
						*&25.612&0.669&15.228&1.171&-1.091\\
						*&*&38.619&24.581&-1.167&-0.915\\
						*&*&*&49.145&0.251&-0.958\\
						*&*&*&*&22.323&3.244\\
						*&*&*&*&*&25.158
						\end{array}} \right]$, $s_i=2$, satisfying~\eqref{cond22}.
					
					With the choice of~\eqref{sim_func} for $V_{i,s_i}$, we get $\tau_i\leq \max \{ \frac{{{\lambda _{\max }}({M_{i,s_i}})}}{{{\lambda _{\min }}({M_{i,s_i}})}}\}$ for $s_i=1,2$. Thus, $1\leq\tau_i\leq 67.61$. Therefore, the parameters in Definition~\ref{SIM_vi_existence} satisfying Assumption~\ref{ass_external_gains} are as $\alpha_i=1$, $\lambda_i=\hat{\kappa_i}\in \left[ {0.3239,0.99} \right]$, $\epsilon_i=1$, $\rho_{i,\rm int}\leq 0.321$, and $\rho_{\rm{i,ext}}\leq 512.312$.
					Recalling the circular interconnection topologies, each subsystem is directly fed by one other subsystem at each time instant.
					Thus, \eqref{gamma} gives $\gamma_{ij}=\tau_i(1+\epsilon_i+\frac{1}{\epsilon_i})\mathop{\max }\limits_{{s_i}} \{|\sqrt{M_{i, s_i}}D_{i, s_i}|^2\}{\bar N_i}\frac{1}{\alpha_j}$
					for $j\in I_{i.s_i}^{\mathrm{in}}$ and $\gamma_{ij}=0$ for $j \notin {I_{i.s_i}^{\mathrm{in}}}$.
					Then, we get
					\begin{align*}
						r(\Psi) &<\sup_{j \in \N} \sum_{i=1}^{\infty} \psi_{ij}<(1+\epsilon_i+\frac{1}{\epsilon_i})\mathop{\max }\limits_{{s_i}} \{|\sqrt{M_{i, s_i}}D_{i, s_i}|^2\}\frac{
							\tau_i}{\hat\kappa_i}\\&\leq 0.991,
					\end{align*}
					which implies the satisfaction of Assumption \ref{ass:spectral-radius} on the spectral radius condition. 
					Therefore, all the hypotheses of Theorem~\ref{MTC} are satisfied.
					
					The norm of the overall error between the output trajectories of the \emph{abstract} and \emph{concrete} systems for three different sizes of networks are shown in Figures~\ref{xx}.
					From the choice of $\hat u$ and stabilizability of $\hat\Sigma$ at the origin, $\lim_{k\to\infty} |\hat{\mathbf u}(k)|_2 \to 0$.
					This together with~\eqref{sim3} implies that the mismatch between output trajectories converges to zero, illustrated by Figure \ref{xx}.
					The reference signals of DGUs are set as $\mathbf{y}_{i,ref}=\left[ {0.8,0.2} \right]^\top$, $i\in\N$.
					The closed-loop output trajectories of the \emph{concrete} subsystems in a set-point tracking  scenario are depicted by Figure~\ref{xx2} in per unit system.
					From Figure~\ref{xx2}, one can see that the overall behavior of the network remains almost identical, though the network size grows dramatically.
					This admits that performances indices are independent of the network size.
					
					\begin{figure*}
						\centering
						\begin{subfigure}[b]{0.65\columnwidth}
							\includegraphics[width=\textwidth,height=4cm]{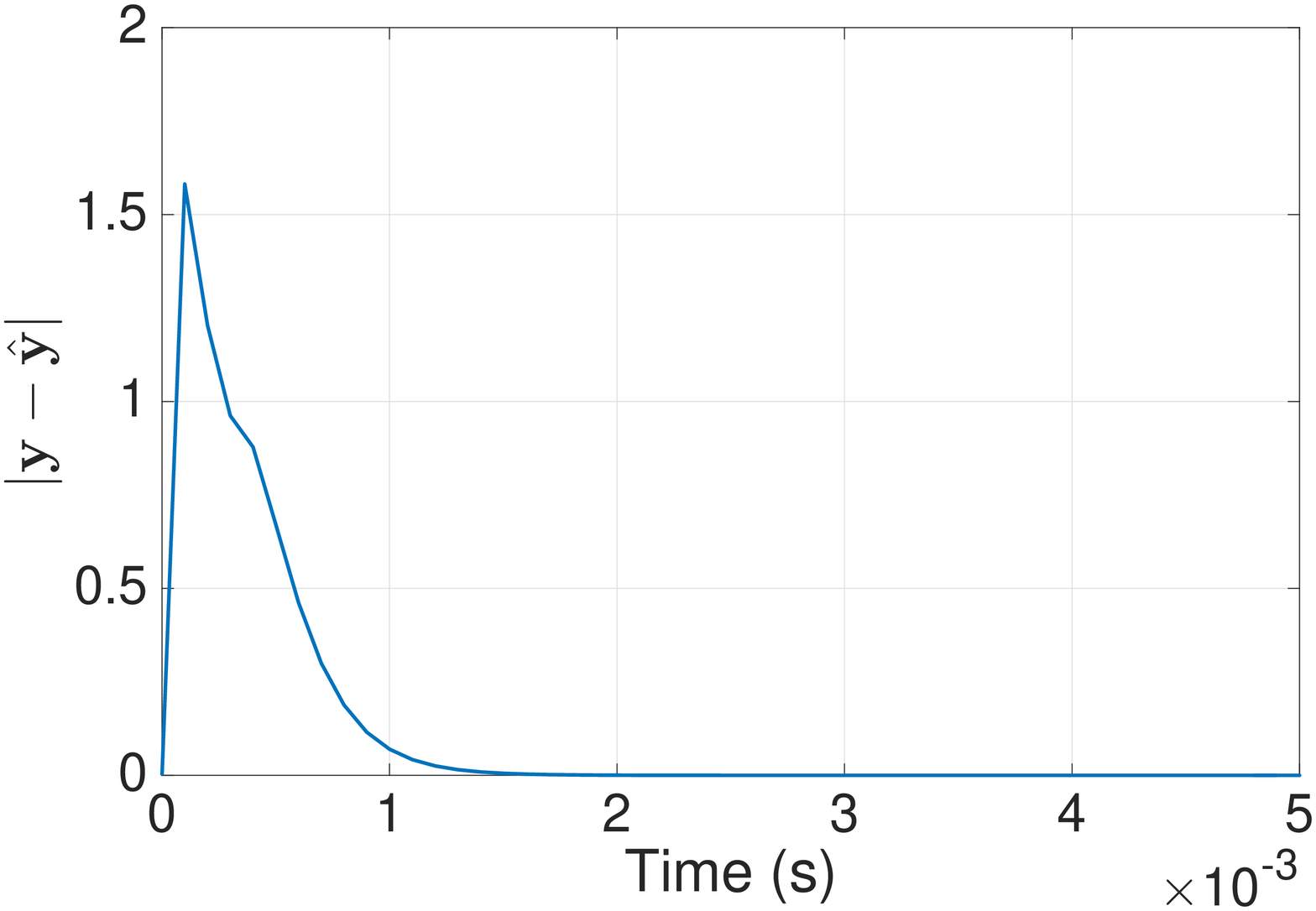}
							\caption{$i={1,\dots,10^2}$}
							\label{60err}
						\end{subfigure}
						\begin{subfigure}[b]{0.65\columnwidth}
							\includegraphics[width=\textwidth,height=4cm]{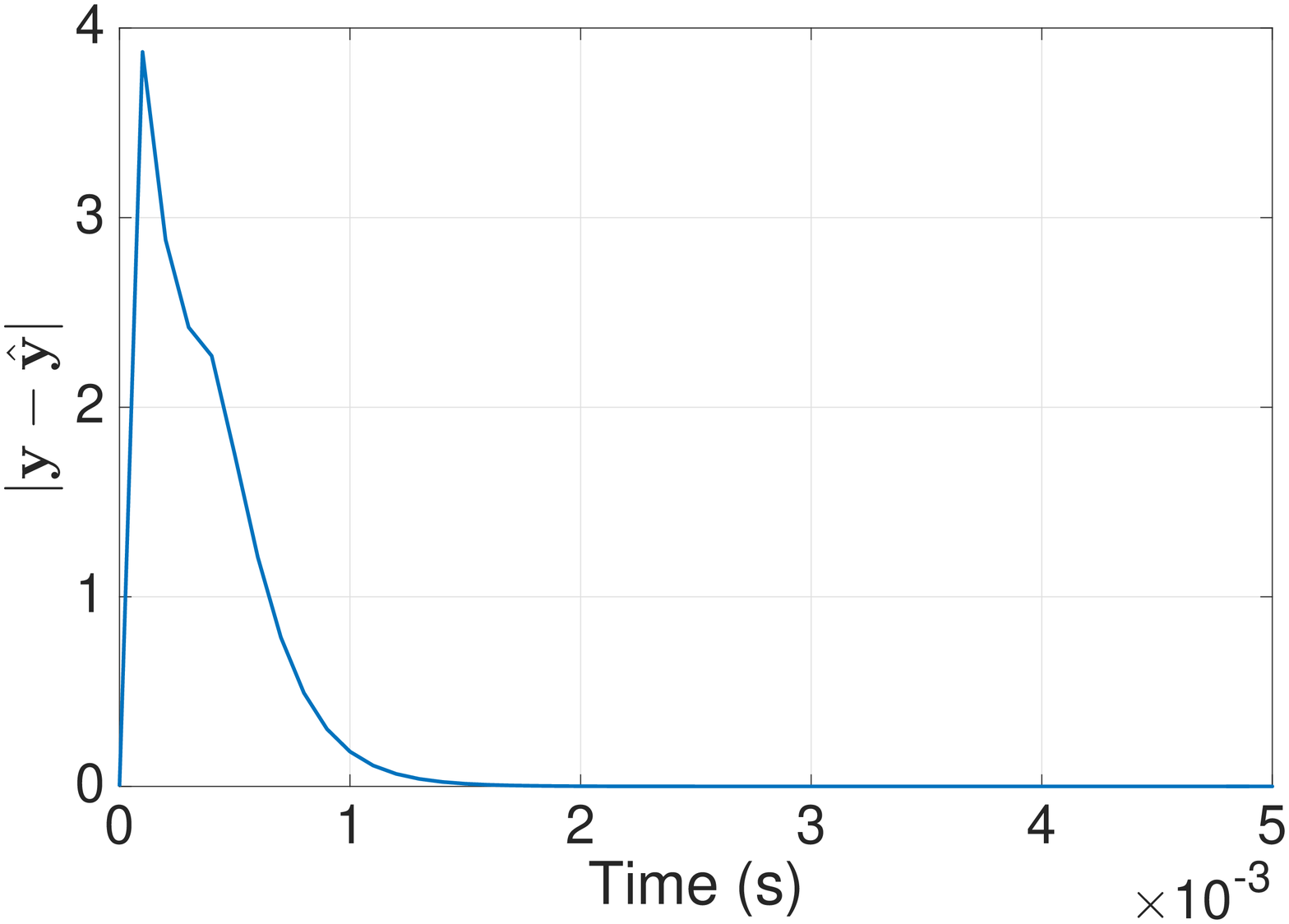}
							\caption{$i={1,\dots,10^3}$}
							\label{600err}
						\end{subfigure}
						\begin{subfigure}[b]{0.65\columnwidth}
							\includegraphics[width=\textwidth,height=4cm]{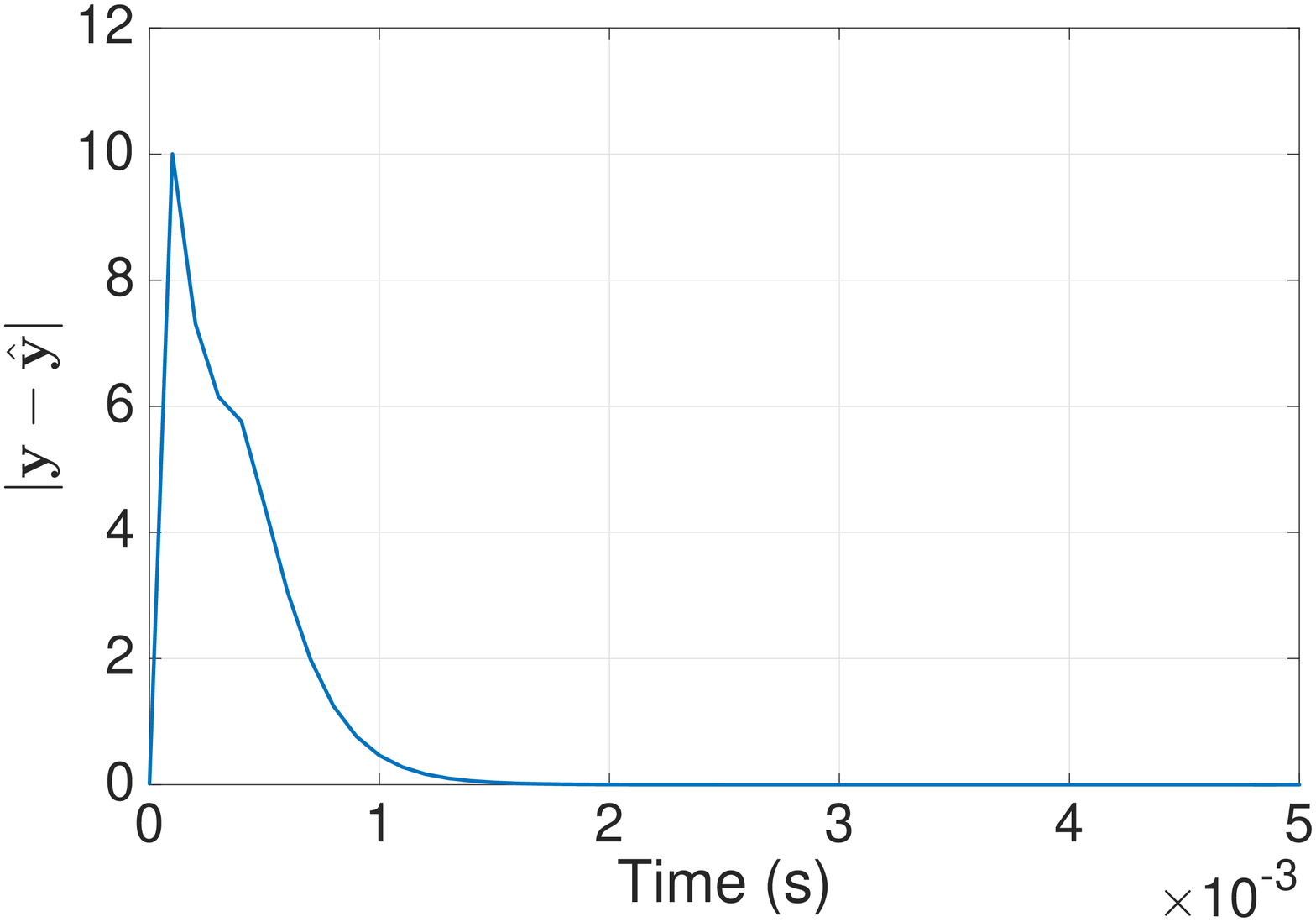}
							\caption{$i={1,\dots,10^4}$}
							\label{6000err}
						\end{subfigure}
						\caption{The error norm between the output trajectories of $\Sigma$ and $\hat \Sigma$ in per unit system.}\label{xx}
					\end{figure*}

					\begin{figure*}
						\centering
						\begin{subfigure}[b]{0.65\columnwidth}
							\includegraphics[width=\textwidth,height=4cm]{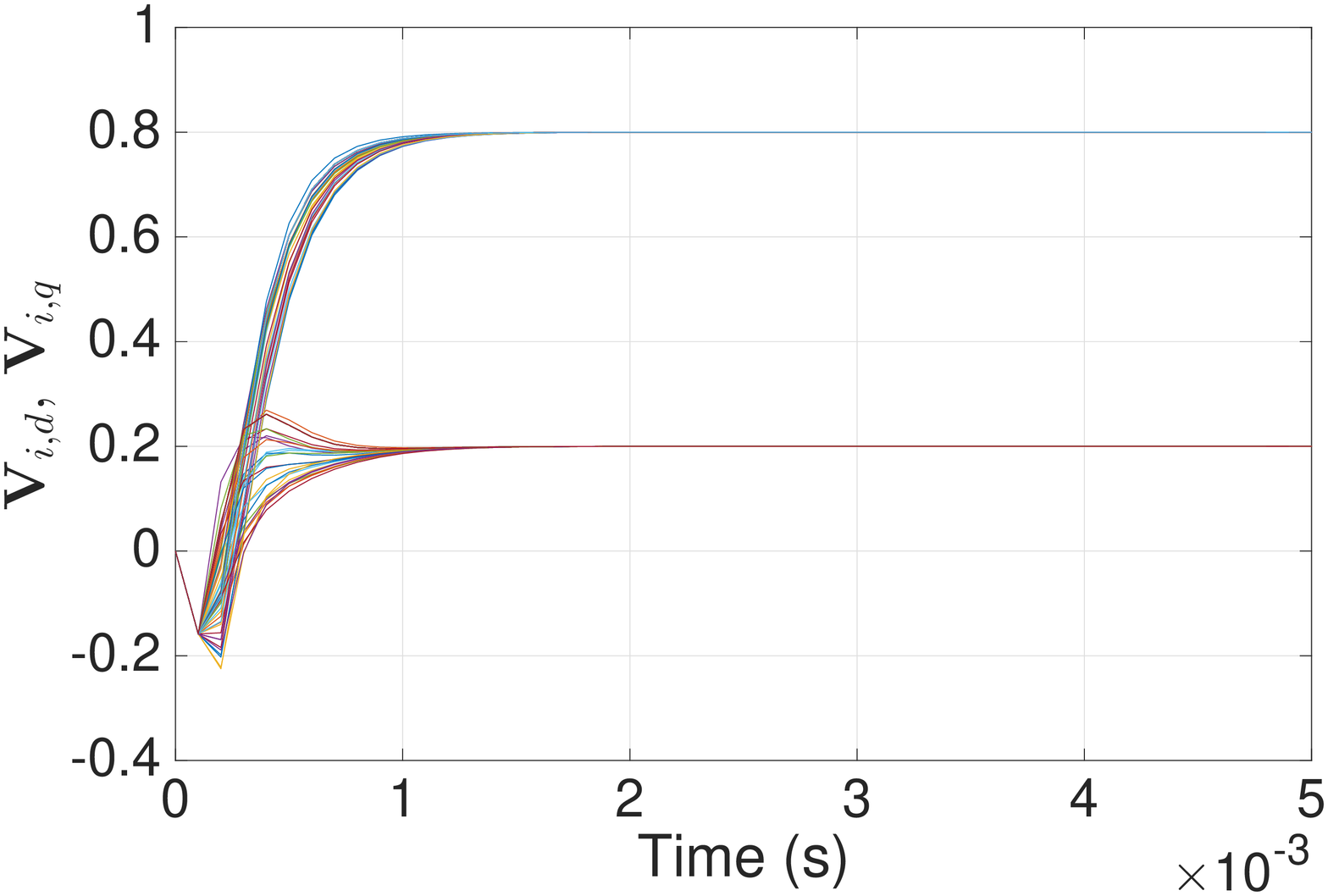}
							\caption{$\i={1,\dots,10^2}$}
							\label{60sd}
						\end{subfigure}
						\begin{subfigure}[b]{0.65\columnwidth}
							\includegraphics[width=\textwidth,height=4cm]{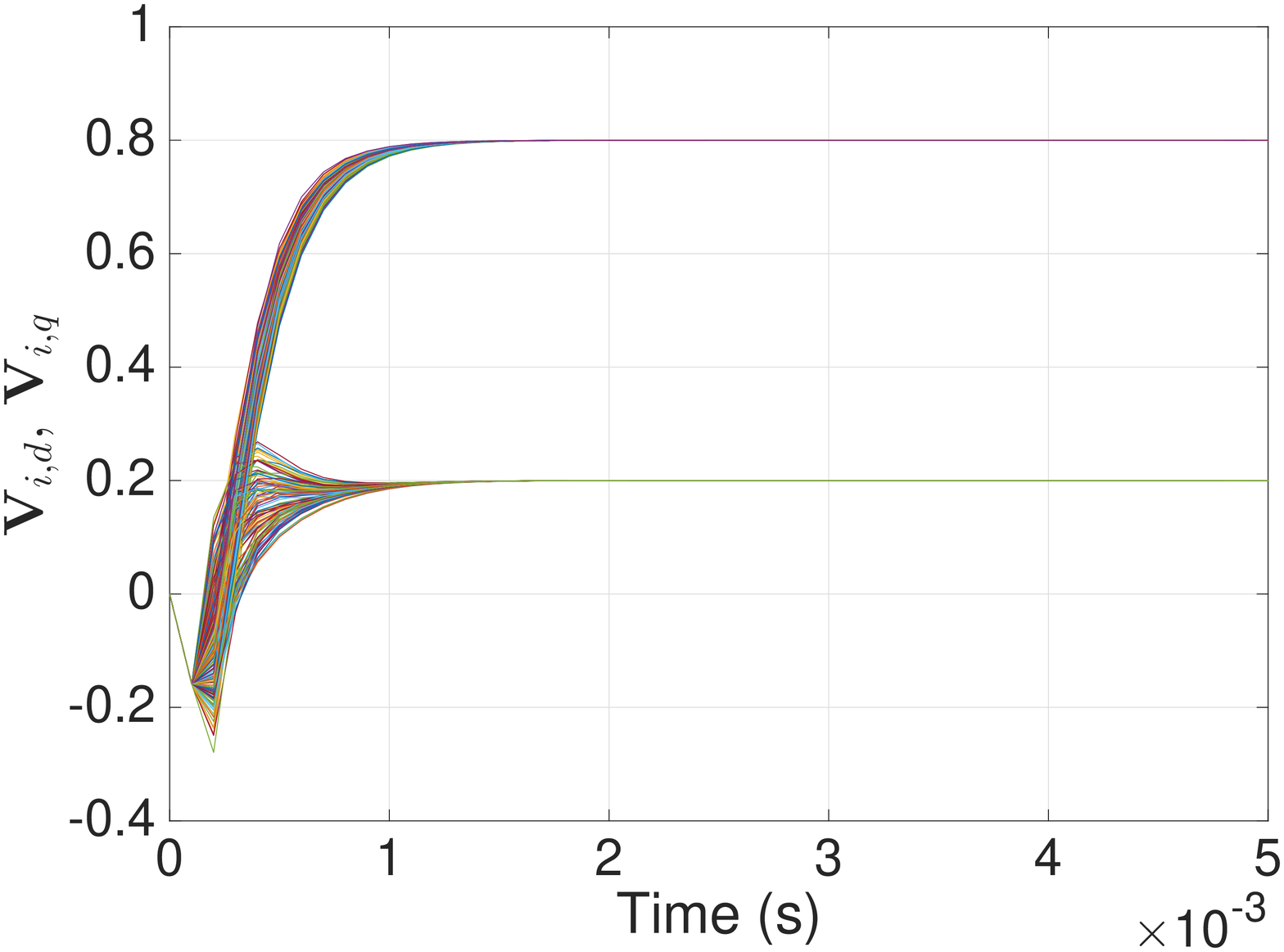}
							\caption{$i={1,\dots,10^3}$}
							\label{600sd}
						\end{subfigure}
						\begin{subfigure}[b]{0.65\columnwidth}
							\includegraphics[width=\textwidth,height=4cm]{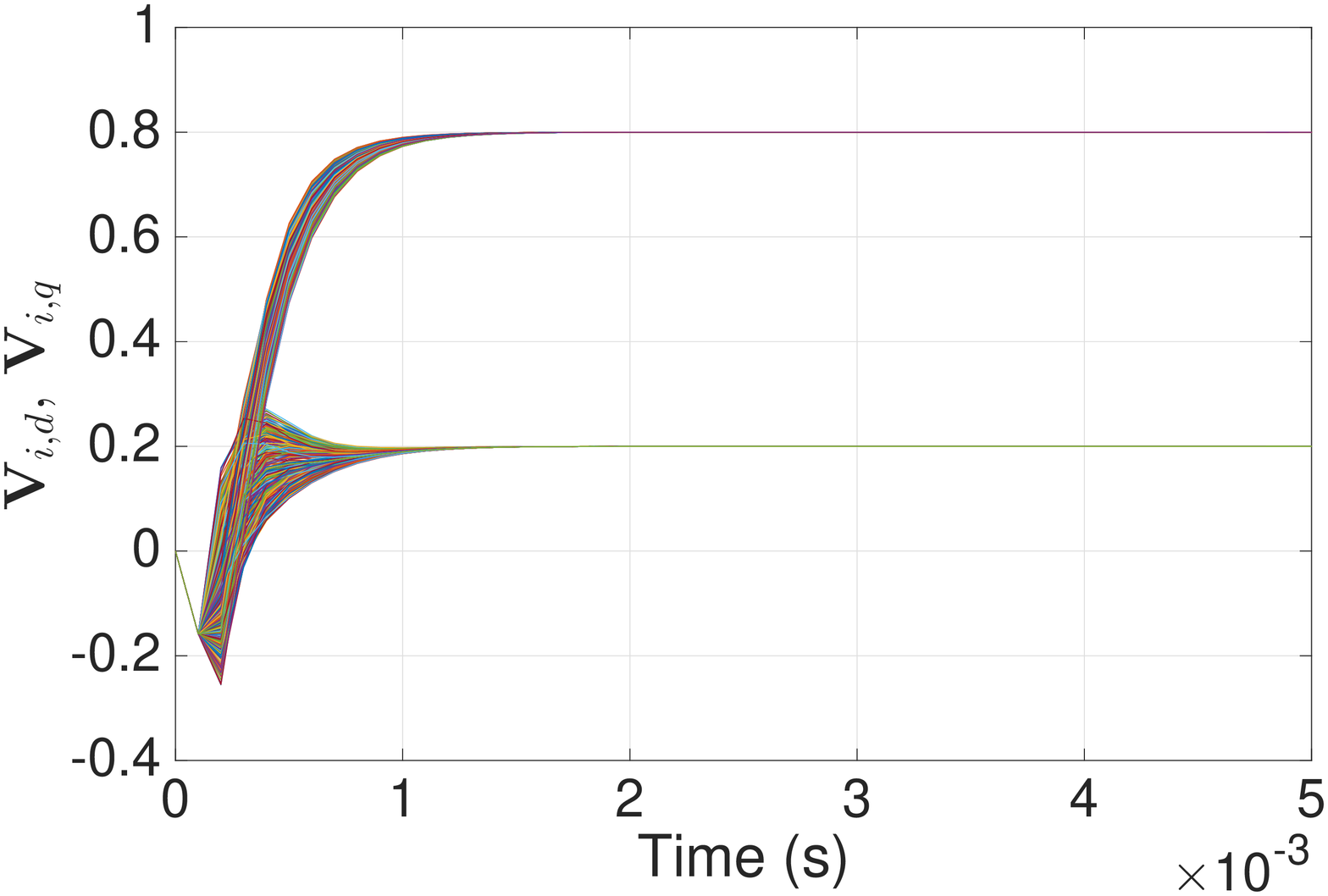}
							\caption{$i={1,\dots,10^4}$}
							\label{6000sd}
						\end{subfigure}
						\caption{The external outputs $\mathbf{V}_{i,d},\;\mathbf{V}_{i,q}$ in per unit system.  }\label{xx2}
					\end{figure*}
					
					\section{Conclusions}\label{sec:Conclusions}
					
					We proposed a compositional approach on the construction of continuous abstractions for infinite networks of switched discrete-time systems with arbitrary switching signals.
					To do this,	we extended the notion of simulation functions to infinite-dimensional systems (networks of infinitely many finite-dimensional switched systems).
					Following the compositionality approach, we assigned to each subsystem an individual simulation function and constructed its local abstraction accordingly.
					Finally, we composed local abstractions to provide an abstraction of the overall network.
					We showed that the aggregation yields a continuous abstraction of the overall concrete network if a small-gain condition, expressed in terms of a spectral radius criterion, is satisfied.
					We also established that our result leads to scale-free compositional method for any finite-but-arbitrarily large networks.
					For linear systems, our conditions for constructing local abstractions boil down to  some linear matrix inequality conditions which can be computed efficiently.
					We applied our results to AC islanded microgrids under switched topologies and showed the scale-freeness of our proposed approach.
					

					\bibliographystyle{IEEEtran}
					\bibliography{referencesj}   
					
				\end{document}